\documentclass[journal]{IEEEtran}
\usepackage[colorlinks=true,linkcolor=blue,urlcolor=blue,citecolor=red,anchorcolor=green]{hyperref}
\usepackage[ascii]{inputenc}
\usepackage{bbm,braket,microtype,aliascnt,mathrsfs,amsmath,amssymb,color,amsthm,graphicx,cleveref}
\usepackage[T1]{fontenc}
\usepackage[english]{babel}

\crefname{section}{Section}{Sections}
\crefname{appendix}{Appendix}{Appendices}
\newtheorem{theorem}{Theorem}\crefname{theorem}{Theorem}{Theorems}
\newtheorem{lemma}[theorem]{Lemma}\crefname{lemma}{Lemma}{Lemmas}
\crefname{corollary}{Corollary}{Corollaries}
\newtheorem{proposition}[theorem]{Proposition}\crefname{proposition}{Proposition}{Propositions}
\theoremstyle{definition}
\crefname{definition}{Definition}{Definitions}
\crefname{remark}{Remark}{Remarks}

\DeclareMathOperator{\tr}{tr}

\DeclareMathOperator{\Prob}{Prob}

\newcommand{\EE}{\mathbb E}
\newcommand{\Hmin}{H_{\min}}
\newcommand{\Hmax}{H_{\max}}
\newcommand{\eps}{\varepsilon}
\allowdisplaybreaks[4]

\begin{document}
%
%
%

\title{Entanglement-assisted capacities of compound quantum channels}
\author{Mario~Berta,
        Hrant~Gharibyan,
        and~Michael~Walter
\thanks{%
M.~Berta acknowledges funding by the SNSF through a fellowship, funding by the Institute for Quantum Information and Matter (IQIM), an NSF Physics Frontiers Center (NFS Grant PHY-1125565) with support of the Gordon and Betty Moore Foundation (GBMF-12500028), and funding support form the ARO grant for Research on Quantum Algorithms at the IQIM (W911NF-12-1-0521).
H.~Gharibyan and M.~Walter gratefully acknowledge support from FQXI and the Simons Foundation.
M.~Walter also gratefully acknowledges support from AFOSR grant FA9550-16-1-0082.}
\thanks{M.~Berta is with the Institute for Quantum Information and Matter, California Institute of Technology, Pasadena, CA 91125 USA (e-mail: berta@caltech.edu).}%
\thanks{H.~Gharibyan is with the Stanford Institute for Theoretical Physics, Stanford University, CA 94305 USA (email: hrant@stanford.edu).}%
\thanks{M.~Walter is with the Stanford Institute for Theoretical Physics, Stanford University, CA 94305 USA (email: michael.walter@stanford.edu).}%
\thanks{\smallskip}
}
\hypersetup{pdftitle={Entanglement-assisted capacities of compound quantum channels},pdfauthor={Mario Berta, Hrant Gharibyan, and Michael Walter}}


\maketitle

\begin{abstract}
We study universal quantum codes for entanglement-assisted quantum communication over compound quantum channels. In this setting, sender and receiver do not know the specific channel that will be used for communication, but only know the set that the channel is selected from. We investigate different variations of the problem: uninformed users, informed receiver, informed sender, and feedback assistance.
We derive single-letter formulas for all corresponding channel capacities. Our proofs are based on one-shot decoupling bounds and properties of smooth entropies.
\end{abstract}





\section{Introduction}\label{sec:introduction}

\IEEEPARstart{S}{hannon's} noisy channel coding theorem establishes that the asymptotic capacity of a fixed independent and identically distributed (IID) channel $\mathcal{J}_{X\to Y}$ is given by the mutual information between the input and the output of the channel~\cite{shannon48},
\begin{align}\label{eq:shannon}
C(\mathcal{J}_{X\to Y})=\sup_P I(X:Y)_{\mathcal{J}(P)}\,,
\end{align}
optimized over all channel inputs. Moreover, it is known that the capacity is neither increased by feedback assistance from the receiver to the sender nor by shared randomness assistance between the sender and the receiver~\cite{shannon48}.\footnote{In fact, not even post-classical assistance like entanglement or general non-signaling correlations increase the capacity~\cite{bennett02,matthews2014finite}.} Extending Shannon's seminal result~\eqref{eq:shannon}, there has been lots of work concerning scenarios where the channel is not fully known but there is some channel uncertainty, see, e.g., the review article~\cite{lapidoth98}. Here, we are interested in the so-called compound setting when the sender and receiver do not know the specific channel that will be used for communication, but only know the set that the channel is selected from. For compound channels $\Lambda_{X \to Y}=\{\mathcal{J}^i_{ X\to Y}\}_{i\in I}$ with index set $I$, the asymptotic IID capacity was determined to~\cite{blackwell59,wolfowitz59},
\begin{align}\label{eq:classical_plain}
C(\Lambda_{X \to Y})=\sup_{P_X}\inf_{i \in I}I(X:Y)_{\mathcal J^i(P)}\,.
\end{align}
Furthermore, and in contrast to the case of a single channel, feedback from the receiver to the sender improves this to~\cite{wolfowitz78,Shrader09},
\begin{align}\label{eq:classical_IS}
C_{F}(\Lambda_{X\to Y})=\inf_{i \in I} C(\mathcal{J}^i_{X\to Y})\,.
\end{align}

Now, unlike in the classical case, IID quantum channels have various different asymptotic capacities and only in special cases formulas to compute these are known~\cite{holevo98,bennett02,schumacher97,devetak05b,lloyd97,shor02}. For the entanglement-assisted quantum capacity, however, a formula similar to~\eqref{eq:shannon} is available and for that reason it is considered to be the natural quantum analogue of the capacity of classical channels. For a fixed quantum channel $\mathcal{N}_{A\to B}$, the entanglement-assisted quantum capacity is given by the quantum mutual information~\cite{bennett02},
\begin{align}\label{eq:Q_entanglement}
Q_E(\mathcal{N}_{A\to B})=\sup_\rho \frac{1}{2}I(A':B)_{\mathcal{N}(\rho)}\,,
\end{align}
optimized over all purified input distributions $\rho_{A'A}$. Like in the classical case, feedback assistance from the receiver to the sender does not increase this capacity~\cite{bowen02}.\footnote{As shown in~\cite{matthews14}, the quantum capacity assisted by non-signaling correlations is also given by~\eqref{eq:Q_entanglement}.} By quantum teleportation~\cite{PhysRevLett.70.1895} and superdense coding~\cite{PhysRevLett.69.2881}, the corresponding entanglement-assisted classical capacity is the same up to a factor of two.

In this paper we study the entanglement-assisted capacities of compound quantum channels $\Pi_{A\to B}=\{\mathcal{N}^i_{A\to B}\}_{i \in I}$, where each $\mathcal{N}^i_{A\to B}$ is a quantum channel between finite-dimensional input and output spaces $A$ and $B$, respectively, and where $I$ is an arbitrary index set. In full analogy to the classical case~\eqref{eq:classical_plain}--\eqref{eq:classical_IS}, we show that the compound quantum capacity without feedback is given by
\begin{align}\label{eq:compound_plain}
Q_{E}(\Pi_{A\to B})=\sup_\rho\inf_{i \in I}\frac{1}{2}I(A':B)_{\mathcal{N}^i(\rho)}\,,
\end{align}
and that feedback improves the capacity to
\begin{align}\label{eq:compound_feedback}
Q_{E,F}(\Pi_{A\to B})=\inf_{i \in I}Q_E(\mathcal{N}^i_{A\to B})\,.
\end{align}
Here, the first formula holds for general compounds whereas we can only show the second formula for finite compound channels, $\lvert I \rvert<\infty$.\footnote{The feedback capacity of general classical compound channels has only been determined recently~\cite{Shrader09}.} In the process of deriving~\eqref{eq:compound_plain}--\eqref{eq:compound_feedback} we also determine the entanglement-assisted compound capacities when either the sender or the receiver is informed about the channel used for communication (this is again in full analogy to the classical case):
\begin{align*}
Q_{E,IR}(\Pi_{A\to B})&=Q_{E}(\Pi_{A\to B})\quad\text{as well as} \\
\quad Q_{E,IS}(\Pi_{A\to B})&=Q_{E,F}(\Pi_{A\to B})\,.
\end{align*}

\paragraph*{Prior work}
The plain classical and quantum capacities of compound quantum channels were studied in~\cite{boche09,hayashi09b,Boche13} and~\cite{bjelakovic2008quantum,bjelakovic09,bjelakovic2009entanglement}, respectively. In~\cite{1751-8121-40-28-S20,Datta2008} the plain and entanglement-assisted classical capacities of a class of quantum channels with long-term memory were determined. As we will see the latter result is equivalent to formula~\eqref{eq:compound_plain} for finite compounds.

\paragraph*{Techniques}
For our proofs we follow the decoupling approach to quantum information theory~\cite{horodecki05,horodecki06,hayden08,abeyesinghe09,dupuis2009decoupling}. In particular we make use of one-shot decoupling bounds~\cite{dupuis2014one} in terms of smooth entropies~\cite{renner2005security,tomamichel2015quantum}. We also utilize various techniques developed in aforementioned references about compound quantum channels.

\section{Overview of results}\label{sec:results}

Here we introduce the setup more formally and give a summary of our results and methods. Let $\Pi_{A\to B} = \{ \mathcal N^i_{A\to B} \}_{i \in I}$ be a finite-dimensional compound quantum channel (i.e., a collection of quantum channels with fixed and finite-dimensional input and output systems, labeled by some index set $I$). Our goal is the quantification of the information-theoretic power of such channels. In particular, we determine their asymptotic IID capacities when free entanglement-assistance between the sender and the receiver in the form of maximally entangled states $\Phi^+_{A_1B_1}$ is available.
\smallskip

\paragraph{Uninformed users} An entanglement transmission code for $\Pi_{A \rightarrow B}$ is given by a quadruple $\{M_0,M_1,\mathcal{E},\mathcal{D}\}$, where $M_0$ is the local dimension of a maximally entangled state $\Phi^+_{A_0R}$ that is to be transmitted, $M_1$ denotes the local dimension of the entanglement-assistance $\Phi^+_{A_1B_1}$, and the quantum channels $\mathcal{E},\mathcal{D}$ are the encoding and decoding operations, respectively.
We say that a triple $(R,n,\delta)$ is achievable for $\Pi_{A \rightarrow B}$ if there exists an entanglement transmission code with
\begin{equation}\label{eq:coding}
\begin{aligned}
&\frac{1}{n}\log M_0\geq R\quad\mathrm{and}\quad\inf_{i\in I}F(\Phi^+_{A_0R},\mathcal{D}_{BB_1\to A_0}\circ\\
&\qquad\mathcal{N}^i_{A\rightarrow B}\circ\mathcal{E}_{A_0A_1\to A}(\Phi^+_{A_0R}\otimes\Phi^+_{A_1 B_1}))\geq1-\delta.
\end{aligned}
\end{equation}
Here, $R$ is the rate of the code, $n\in\mathbb{N}$ the number of channel uses, and $\delta>0$ the tolerated error measured in terms of Uhlmann's fidelity $F(\rho,\sigma):=\lVert\sqrt{\rho}\sqrt{\sigma}\rVert_1^2$. This means that for every channel $\mathcal{N}^i_{A\to B}$ in the compound, the fidelity of transmission should be high, averaged over a uniform message of size $M_0$. In the asymptotic IID limit the entanglement-assisted capacity of the compound quantum channel $\Pi_{A \rightarrow B}$ is then defined as
\begin{equation}\label{eq:capacity}
\begin{aligned}
&Q_E(\Pi_{A \rightarrow B}):= \lim_{\delta\to0}\liminf_{n\to\infty}\\
&\quad \sup\left\{R:\text{$(R,n,\delta)$ is achievable for $\Pi_{A \rightarrow B}$}\right\}.
\end{aligned}
\end{equation}
Slightly different definitions of entanglement transmission codes~\eqref{eq:coding} and, correspondingly, channel capacities~\eqref{eq:capacity} are possible as well and we point to~\cite{Kretschmann04} for an in-depth discussion. We prove the following formula for $Q_E(\Pi_{A \rightarrow B})$:

\begin{theorem}\label{thm:uninformed}
Let $\Pi_{A \rightarrow B}  = \{ \mathcal N^i_{A\rightarrow B} \}_{i \in I}$ be a compound channel with arbitrary index set $I$.
Then, we have
\begin{align}\label{eq:Q_E}
Q_E(\Pi_{A \rightarrow B}) = \frac12 \sup_{\rho} \inf_{i\in I}I(A':B)_{\mathcal{N}^i(\rho)},
\end{align}
where the supremum is over all purified input distributions $\rho_{A'A}$ with $A \cong A'$.
\end{theorem}

To prove the converse, we rely on the quantum generalization of the meta converse for channel coding from~\cite{matthews2014finite}, together with some extensions which were inspired by the work~\cite{Datta2008}.
To prove achievability, our starting point is a one-shot coding theorem for fixed channels~\cite{dupuis2009decoupling} that is formulated in terms of smooth entropies. These entropies were introduced to study structureless (non-IID) resources~\cite{renner2005security} and have many properties similar to the von Neumann entropy~\cite{tomamichel2015quantum}. By applying the one-shot result to the average channel $\overline\Pi_{A\to B}:= \frac 1 {\lvert I \rvert} \sum_{i\in I} \mathcal{N}^i_{A\to B}$ associated with a finite compound channel (cf.~the work~\cite{bjelakovic2008quantum})
and exploiting basic properties of smooth entropies, we obtain a one-shot coding theorem for finite compound channels. In the asymptotic limit, this establishes \cref{thm:uninformed} for finite compound channels. In light of this proof and the equivalence between the entanglement-assisted transmission of classical and quantum information, as discussed below, this result can be seen to be equivalent to the work~\cite{Datta2008} where the entanglement assisted classical capacity of a class of quantum channels of the form $\hat{\mathcal N}_{A^n\to B^n}=\sum_{i \in I} p_i \, (\mathcal{N}^i_{A\to B})^{\otimes n}$ with $\{p_i\}_{i \in I}$ a finite probability distribution was determined (using different techniques). To extend the proof to arbitrary compound channels we use a discretization argument via net techniques (similarly to what is done in, e.g., \cite{bjelakovic2009entanglement}). The basic reason why this works is that the random coding schemes that we make use of have a well behaved error dependence, together with the mutual information being nicely continuous.

\paragraph{Informed receiver} When the receiver but not the sender knows which one of the channels $\mathcal{N}_{A\to B}^i$ from the compound will be used for the information transmission, the decoder $\mathcal D^i_{BB_1\to A_0}$ can depend on the channel and thus~\eqref{eq:coding} becomes
\begin{align*}
\frac{1}{n}\log M_0\geq R\quad\mathrm{and}\quad\inf_{i\in I}F(\Phi^+_{A_0R},\mathcal{D}^i_{BB_1\to A_0}\circ \\
\quad \mathcal{N}^i_{A\rightarrow B}\circ\mathcal{E}_{A_0A_1\to A}(\Phi^+_{A_0R}\otimes\Phi^+_{A_1 B_1}))\geq1-\delta,
\end{align*}
The corresponding capacity $Q_{E,IR}$ is then defined as in~\eqref{eq:capacity} and we find that it does not increase compared to the uninformed case.

\begin{theorem}\label{thm:IR}
Let $\Pi_{A \rightarrow B}  = \{ \mathcal N^i_{A\rightarrow B} \}_{i \in I}$ be a compound channel with arbitrary index set $I$. Then, we have
\begin{align*}
Q_{E, IR}(\Pi_{A \rightarrow B}) = Q_E(\Pi_{A \rightarrow B}).
\end{align*}
\end{theorem}

For the proof we simply note that the converse established for the uninformed case does not rely on the decoder being uninformed and thus holds verbatim for informed receivers.

\smallskip

\paragraph{Informed sender} When the sender but not the receiver is aware of which channel $\mathcal{N}_{A\to B}^i$ from the compound will be used for the information transmission, the encoder $\mathcal{E}^i_{A_0A_1\to A}$ may now depend on $i\in I$ and thus~\eqref{eq:coding} becomes
\begin{align*}
&\frac{1}{n}\log M_0\geq R\quad\mathrm{and}\quad\inf_{i\in I}F(\Phi^+_{A_0R},\mathcal{D}_{BB_1\to A_0}\circ \\
&\qquad \mathcal{N}^i_{A\rightarrow B}\circ\mathcal{E}^i_{A_0A_1\to A}(\Phi^+_{A_0R}\otimes\Phi^+_{A_1 B_1}))\geq1-\delta.
\end{align*}
The corresponding capacity $Q_{E,IS}$ is then defined as in~\eqref{eq:capacity} and we find that in general the capacity increases. More precisely, we show that infimum and supremum in the formula~\eqref{eq:Q_E} can be interchanged, so that the entanglement-assisted quantum capacity of a compound equals the infimum of the capacities of its constituents:

\begin{theorem}\label{thm:IS}
Let $\Pi_{A\to B}=\{\mathcal{N}^i_{A\to B}\}_{i \in I}$ be a compound channel with arbitrary index set $I$. Then, we have
\begin{align*}
Q_{E,IS}(\Pi_{A\to B})&=\inf_{i\in I}Q_E(\mathcal{N}^i_{A\to B}),
\intertext{where}
Q_E(\mathcal{N}^i_{A \to B})&=\frac12\max_{\rho}I(A':B)_{\mathcal{N}^i(\rho)}
\end{align*}
is the entanglement-assisted capacity of the $i$-th channel as determined in~\cite{bennett02}.
\end{theorem}

The converse direction follows directly from the converse for a fixed channel~\cite{bennett02}, since the capacity can not be higher than the capacity of any single channel in the compound. For the achievability, we derive new one-shot decoupling bounds that then imply the existence of a universal decoder by a standard argument via Uhlmann's theorem. As before we reduce from general to finite compounds by using a discretization argument via net techniques.

\smallskip

\paragraph{Feedback} We also solve the scenario with free feedback from receiver to sender.\footnote{Since we assume free entanglement-assistance, classical and quantum feedback are equivalent by quantum teleportation~\cite{PhysRevLett.70.1895} and superdense coding~\cite{PhysRevLett.69.2881}.} The first round of a feedback-assisted entanglement transmission code starts with a quantum or classical feedback message $X_{A}^{(0)}$ that is correlated with some $X_{B}^{(0)}$ at the receiver, followed by an encoder $\mathcal{E}_{X_{A}^{(0)}A_0A_1\to A}$, the channel $\mathcal{N}^i_{A\rightarrow B}$ from the compound, and a decoder $\mathcal{D}_{X_{B}^{(0)}BB_1\to A_0}$. Now, for $n$ channel uses the procedure repeats $n$-times.  The corresponding capacity $Q_{E,F}$ is then defined as in~\eqref{eq:capacity}, where we have an application of $n$-rounds of encoders and decoders and the fidelity bound in~\eqref{eq:coding} holds for these $n$-rounds~\cite{bowen02}. We find that the corresponding feedback capacity is the same as in the case of informed senders (at least for finite compounds). Therefore, in analogy to the classical case, feedback in general increases the entanglement-assisted quantum capacity of compound channels.

\begin{theorem}\label{thm:feedback}
Let $\Pi_{A\to B}=\{\mathcal{N}^i_{A\to B}\}_{i\in I}$ be a compound channel with finite index set $|I|<\infty$. Then, we have
\begin{equation*}
Q_{E,F}(\Pi_{A\to B})=Q_{E,IS}(\Pi_{A\to B}).
\end{equation*}
\end{theorem}

The converse direction again follows from the converse of the corresponding result for a single channel~\cite{bowen02}. For the achievability we make use of the channel estimation techniques from~\cite{1751-8121-40-28-S20}. In particular, we use the first $\sqrt{n}$ instances of the channel to estimate the channel on the receiver's side and then use feedback to transfer the channel index $i\in I$ to the sender. This allows to use the informed sender protocol for the remaining $n-\sqrt{n}$ channel uses and leads to the same capacity as in the informed sender case. Unfortunately, it is unclear how to apply this technique for general compounds since in this case the trade-off between the discretization and the channel estimation parameters might not scale well enough. We note that even for classical compound channels this issue has only been resolved relatively recently~\cite{Shrader09}, and we leave it as an open problem in the quantum case.

\smallskip

\paragraph{Application} We consider a situation where some channel $\mathcal{N}_{A\to B}$ used for information transmission has only been characterized up to some precision $\eps>0$.\footnote{More precisely, we assume that for some fixed input system we know the channel $\mathcal{N}_{A\to B}$ up to $\eps>0$ in 
the diamond norm (see \cref{sec:notation} for its definition).} This corresponds to uninformed users (\cref{thm:uninformed}) and by the continuity of the conditional entropy~\cite{Alicki04,winter15} we find that we can still transmit quantum information at a rate of at least
\begin{align}\label{eq:example_continuity}
Q_E(\mathcal{N}_{A\to B})-\eps\log d_A+\left(1+\frac{\eps}{2}\right)h\left(\frac{\eps}{2+\eps}\right),
\end{align}
where $d_A$ denotes the input dimension of the quantum channel, and $h(p)$ denotes the binary entropy function. Hence, we can conclude that not knowing the quantum channel perfectly only sightly impacts our ability to transmit quantum information. (This does not just follow from the continuity of the entanglement-assisted channel capacity since we need a universal coding scheme that works for all channels in the $\eps$-neighborhood.)
\smallskip

\paragraph{Classical communication}
Lastly, we discuss the entanglement-assisted transmission of classical information through the compound quantum channel $\Pi_{A\to B}$. A classical code for $\Pi_{A \rightarrow B}$ is given by a quadruple $\left\{M_0,M_1,\{\mathcal{E}^m\},\{\Lambda^m\}\right\}$, where $M_0$ is the number of classical messages to be transmitted, $M_1$ denotes the local dimension of the entanglement-assistance $\Phi^+_{A_1B_1}$, $\{\mathcal{E}^m_{A_1\to A}\}_{m=1}^{M_0}$ are the encoding operations for the classical messages, and $\{\Lambda^m_{BB_1} \}_{m=1}^{M_0}$ the POVM elements of the decoding measurement. Then, for example for the uninformed users setting like in~\eqref{eq:coding}, we say that a triple $(R,n,\delta)$ is achievable for $\Pi_{A \rightarrow B}$ if there exists a classical code with
\begin{equation}\label{eq:classical_coding}
\begin{aligned}
&\frac{1}{n}\log M_0\geq R\quad\mathrm{and}\quad\inf_{i\in I} \frac{1}{M_0}\sum_{m=1}^{M_0}\tr\bigl[ \Lambda_{BB_1}^m \\
&\qquad\qquad\left(\mathcal{N}^i_{A\to B}\circ \mathcal{E}^m_{A_1\to A}(\Phi^+_{A_1 B_1})\right)\bigr]\geq1-\delta.
\end{aligned}
\end{equation}
This means that for every possible choice $\mathcal{N}^i_{A\to B}$ in the compound, the probability of success to retrieve the message should be high, averaged uniformly over all messages $m\in M_0$. The corresponding entanglement-assisted classical capacity $C_E$ is then defined in the exact same way as in~\eqref{eq:capacity}. Now, in the presence of free entanglement-assistance we have quantum teleportation~\cite{PhysRevLett.70.1895} and superdense coding~\cite{PhysRevLett.69.2881} available: this allows to transform entanglement transmission codes into classical codes and vice versa. More precisely, following~\cite[(46) \& Appendix B]{matthews14} we get with superdense coding that,
\begin{align*}
&\text{$\{M_0,M_1,\mathcal{E},\mathcal{D}\}$ with fidelity $1-\delta$ } \;\Rightarrow\; \\
&\text{$\left\{M_0^2,M_1M_0,\{\mathcal{E}^m\},\{\Lambda^m\}\right\}$ with success probability $1-\delta$},
\end{align*}
and vice versa with quantum teleportation that,
\begin{align*}
&\text{$\left\{M_0,M_1,\{\mathcal{E}^m\},\{\Lambda^m\}\right\}$ with success probability $1-\delta$}\;\Rightarrow\; \\
&\text{$\{\sqrt{M_0},M_1M_0,\mathcal{E},\mathcal{D}\}$ with fidelity $1-\delta$}.
\end{align*}
Asymptotically this leads to the following channel capacity identities:
\begin{equation}\label{eq:classical vs quantum}
2=\frac{C_E}{Q_E} = \frac{C_{E,IR}}{Q_{E,IR}}=\frac{C_{E,IS}}{Q_{E,IS}}=\frac{C_{E,F}}{Q_{E,F}}.
\end{equation}
\smallskip

\paragraph{Extensions} Our proofs bring up new tools that can deployed for simplifying previous works about capacities of compound quantum channels~\cite{boche09,hayashi09b,Boche13,1751-8121-40-28-S20,bjelakovic2008quantum,bjelakovic09,bjelakovic2009entanglement}. 
In \cref{app:plain IS} we present such an argument for the plain quantum capacity of compound quantum channels.

\smallskip

\paragraph{Related work} After completion of our work, we have learned about the concurrent work~\cite{boche16paper,boche16} that also determines the entanglement-assisted quantum capacity of compound channels for uninformed users (amongst various other results). However, the key techniques used for the proofs differ significantly. The authors of \cite{boche16paper, boche16} use Weyl unitaries as special encoding channels to convert the problem of entanglement-assisted compound channels to a question about classical-quantum compound channels. Then, they deploy established universal codes \cite{boche09} for classical-quantum channels.

\section{Notation}\label{sec:notation}

Here, we introduce our notation and give some standard definitions and lemmas that we will throughout this paper.
For more about quantum information theory we point to the excellent textbooks~\cite{wildebook13,tomamichel2015quantum}.

Let $A,B,C,\ldots$ denote finite-dimensional Hilbert spaces, let $d_A$ denote the dimension of $A$, and let $d_{AB}=d_Ad_B$ denote the dimension of $AB=A\otimes B$. We write $\mathcal P(A)$ for the set of positive semidefinite operators, $\mathcal S(A)$ for the set of normalized quantum states, $\mathcal S_{\leq}(A)$ for the set of sub-normalized quantum states (i.e., positive semidefinite operators of trace no larger than one). We use $\tau_A = I_A / d_A$ for the maximally mixed state on $A$ and $\Phi^+_{AA'} = \frac 1 {d_A} \sum_{i,j} \ket{ii}\!\!\bra{jj}_{AA'}$ for a maximally entangled state, where $A' \cong A$.


\subsection{Distance measures}

\paragraph{For quantum states} We write $\lVert X \rVert_1:=\tr\left[\sqrt{X^\dagger X}\right]$ for the trace norm. The fidelity is defined for arbitrary positive semidefinite operators $\rho, \sigma \in \mathcal P(A)$ by $F(\rho,\sigma):=\lVert\sqrt\rho \sqrt\sigma\rVert_1^2$. Note that for pure states $\rho$ and $\sigma$ we get $F(\rho, \sigma) = \lvert\braket{\rho, \sigma}\rvert^2$. A particular version of Uhlmann's theorem is~\cite[Theorem 3.1]{dupuis2009decoupling},
\begin{equation}\label{eq:uhlmann fred}
\begin{aligned}
  F(\rho_A, \sigma_A)
  &= \max_{V_{B\to C}} F\left(V_{B\to C} \rho_{AB} V_{B\to C}^\dagger, \sigma_{AC}\right) \\
  &= \max_{W_{C\to B}} F\left(\rho_{AB}, W_{C\to B} \sigma_{AC} W_{C\to B}^\dagger\right),
\end{aligned}
\end{equation}
where $\rho_{AB}$ is an arbitrary purification of $\rho_A$, $\sigma_{AC}$ an arbitrary purification of $\sigma_A$, and where the maximizations are over partial isometries. For normalized quantum states $\rho, \sigma \in \mathcal S(A)$, fidelity and trace norm are related by the Fuchs-van de Graaf inequalities:
\begin{align}\label{eq:fuchs-van de graaf}
	\Bigl( 1 - \frac {\lVert \rho - \sigma \rVert_1} 2 \Bigr)^2 \leq F(\rho, \sigma) \leq 1 - \frac14 \lVert \rho - \sigma \rVert_1^2.
\end{align}
As a consequence, one obtains the following lemma, which is central to the decoupling approach to quantum information theory.

\begin{lemma}[{cf.~\cite[Corollary 3.2]{dupuis2009decoupling}}]\label{lem:uhlmann one is normalized}
  Let $\rho_{AB} \in \mathcal P(AB)$ and $\sigma_{AC} \in \mathcal S(AB)$ be pure states with $\lVert \rho_A - \sigma_A \rVert_1 \leq \delta$.
  Then, there exist partial isometries $V_{B\to C}$ and $W_{C\to B}$ such that
  \begin{align*}
    \lVert V_{B\to C} \rho_{AB} V_{B\to C}^\dagger - \sigma_{AC} \rVert_1 &\leq \delta + 2\sqrt{2\delta}
  \intertext{as well as}
    \lVert \rho_{AB} - W_{C\to B} \sigma_{AC} W_{C\to B}^\dagger \rVert_1 &\leq \delta + 2\sqrt{2\delta}.
  \end{align*}
\end{lemma}

A proof can be found in \cref{app:smooth entropies}.
We note that $\delta + 2\sqrt{2\delta} \leq 4\sqrt{\delta}$ if $\delta \leq 1$, which will often be the case.

Fidelity and trace norm can be related for arbitrary positive semidefinite operators $\rho, \sigma \in \mathcal P(A)$~\cite[Lemmas~A.2.4 and A.2.6]{renner2005security},
\begin{align*}
	&\Bigl( \frac {\tr[\rho] + \tr[\sigma]} 2 - \frac {\lVert \rho - \sigma \rVert_1} 2 \Bigr)^2 \leq F(\rho, \sigma) \\
	\leq\;&\frac {\big(\tr[\rho] + \tr[\sigma]\big)^2} 4 - \frac14 \lVert \rho - \sigma \rVert_1^2,
\end{align*}
which generalizes \eqref{eq:fuchs-van de graaf}.
In particular, the lower bound implies that
\begin{equation}\label{eq:fidelity lower bound one is normalized}
  F(\rho, \sigma)\geq 1 - 2 \lVert \rho - \sigma \rVert_1\quad\text{for $\rho \in \mathcal P(A)$, $\sigma \in \mathcal S(A)$.}
\end{equation}
We will also need the following continuity bound for the fidelity, which follows from, e.g., \cite[Lemma B.9]{fawzi2015quantum}: for $\rho, \sigma, \sigma' \in \mathcal S_\leq(A)$ we have
\begin{equation}\label{eq:fidelity_continuity}
\begin{aligned}
  \lvert F(\rho, \sigma) - F(\rho, \sigma') \rvert
  &\leq 2 \lvert \sqrt{F(\rho, \sigma)} - \sqrt{F(\rho, \sigma')} \rvert  \\
  &\leq 2 \sqrt{ \lVert \sigma - \sigma' \rVert_1  }.
\end{aligned}
\end{equation}
Lastly, for $\rho, \sigma \in \mathcal S_\leq(A)$ we need the purified distance $P(\rho,\sigma):= \sqrt{1 - F_*(\rho,\sigma)}$, defined in terms of the generalized fidelity $F_*(\rho,\sigma):= \Big( \sqrt{F(\rho,\sigma)} + \sqrt{(1-\tr[\rho])(1-\tr[\sigma])}\Big)^2$. We have
\begin{align}\label{eq:purified distance and trace norm subnormalized}
\lVert \rho - \sigma \rVert_1 \leq 2 P(\rho, \sigma)\quad\text{for $\rho, \sigma \in \mathcal S_\leq(A)$.}
\end{align}
\smallskip

\paragraph{For quantum operations}
The diamond norm of a super-operator $\mathcal{N}_{A\to B}$ is defined as 
\begin{align*}
\lVert \mathcal N \rVert_\diamond := \max_{\rho_{AA'} \in \mathcal S(AA')} \lVert \mathcal{N}_{A\to B}(\rho_{AA'}) \rVert_1\quad\text{with $A' \cong A$.}
\end{align*}
From the continuity bounds for the fidelity~\eqref{eq:fidelity_continuity} we find for $\mathcal T_{A\to B},\mathcal T'_{A\to B}$ quantum operations (completely positive and trace-nonincreasing maps) and $\rho,\sigma\in\mathcal{S}_\leq(A)$ that,
\begin{equation}\label{eq:fidelity channel continuity}
\begin{aligned}
  \lvert F(\rho, \mathcal T(\sigma)) - F(\rho, \mathcal T'(\sigma)) \rvert
  &\leq 2 \sqrt{ \lVert (\mathcal T - \mathcal T')(\sigma) \rVert_1 } \\
  &\leq 2 \sqrt{ \lVert \mathcal T - \mathcal T' \rVert_\diamond }.
\end{aligned}
\end{equation}

\subsection{Entropies}

\paragraph{Von Neumann entropy} For $\rho_{AB} \in \mathcal S(AB)$ the von Neumann entropy\footnote{All logarithms are to base 2.} of system $A$ is defined by $H(A)_\rho:= -\tr[\rho_A \log \rho_A]$. The conditional entropy of $A$ given $B$ is $H(A|B)_\rho:=H(AB)_\rho-H(B)_\rho$, and the mutual information between $A$ and $B$ is $I(A:B)_\rho:=H(A)_\rho+H(B)_\rho-H(AB)_\rho$. All of these functions are continuous, and we will use the following continuity bounds for $\rho_{AB},\sigma_{AB} \in \mathcal S(AB)$: For the conditional entropy~\cite{Alicki04,winter15}, we have
\begin{equation}\label{eq:fannes cond entropy}
\begin{aligned}
  &\lvert H(A|B)_\rho - H(A|B)_\sigma\rvert\\
  \leq\;&2T\log d_A+(1+T)h\left(\frac{T}{1+T}\right),
\end{aligned}
\end{equation}
while for the mutual information~\cite{shirokov15},
\begin{equation}\label{eq:fannes mutual info}
\begin{aligned}
  &\lvert I(A:B)_\rho - I(A:B)_\sigma\rvert\\
  \leq\;&2T\log\min\{d_A,d_B\}+2(1+T)h\left(\frac{T}{1+T}\right).
\end{aligned}
\end{equation}
In both cases, $T := \frac12 \lVert \rho_{AB} - \sigma_{AB} \rVert_1$, and $h(p)$ denotes the binary entropy function.

\smallskip

\paragraph{Smooth entropies} For $\rho_{AB} \in \mathcal S_\leq(AB)$ the conditional min-entropy of $A$ given $B$ is defined as~\cite{renner2005security}
\begin{equation}\label{eq:min entropy sdp}
\begin{aligned}
&\Hmin(A|B)_\rho\\
:= &-\log \min \{ \tr[\sigma_B] : \sigma_B \in \mathcal P(B), \rho_{AB} \leq I_A \otimes \sigma_B \}.
\end{aligned}
\end{equation}
The $\eps$-smooth conditional min-entropy is defined by the following optimization over nearby states~\cite{tomamichel09}:
\begin{equation}\label{eq:smooth min entropy}
\begin{aligned}
&\Hmin^\eps(A|B)_\rho\\
:=\;&\sup \big\{ \Hmin(A|B)_{\widetilde\rho} : \widetilde\rho \in \mathcal S_\leq(AB), P(\widetilde\rho, \rho) \leq \eps \big\}
\end{aligned}
\end{equation}
Note that we have $\Hmin^0(A|B)_\rho=\Hmin(A|B)_\rho$ since the purified distance defines a metric on the sub-normalized states. The (smooth) conditional max-entropy for $\eps\geq0$ is defined by duality as~\cite{koenig08,tomamichel09},
\begin{align}\label{eq:smooth max entropy duality}
\Hmax^\eps(A|B)_\rho:= -\Hmin^\eps(A|C)_\rho,
\end{align}
where the smooth conditional min-entropy on the right-hand side is evaluated with respect to the reduced density matrix $\rho_{BC}$ of an arbitrary purification $\ket{\rho_{ABC}}$ of $\rho_{AB}$. It holds that~\cite{tomamichel09}
\begin{equation}\label{eq:smooth max entropy}
\begin{aligned}
&\Hmax^\eps(A|B)_\rho \\
=\,&\sup \big\{ \Hmax(A|B)_{\widetilde\rho} : \widetilde\rho \in \mathcal S_\leq(AB), P(\widetilde\rho, \rho) \leq \eps \big\}.
\end{aligned}
\end{equation}
where $\Hmax(A|B) = \Hmax^0(A|B)$.
We will make use of the following lemmas about the smooth conditional min- and max-entropy of convex combination quantum states.

\begin{lemma}\label{lem:min entropy mixture}
  Let $\{ \rho^i_{AB} \}_{i=1}^N \subseteq \mathcal S(AB)$ be quantum states, $\{p_i\}_{i=1}^N$ a probability distribution, and $\rho_{AB} = \sum_{i=1}^N p_i \rho^i_{AB}$.
  Then, we have
	\begin{align*}
		\Hmin^\eps(A|B)_\rho \geq \min_i \Hmin^\eps(A|B)_{\rho^i}.
	\end{align*}
\end{lemma}

\begin{lemma}\label{lem:max entropy mixture}
  Let $\{ \rho^i_{AB} \}_{i=1}^N \subseteq \mathcal S(AB)$ be quantum states, $\{p_i\}_{i=1}^N$ a probability distribution, and $\rho_{AB} = \sum_{i=1}^N p_i \rho^i_{AB}$.
  Then, we have
	\begin{align*}
		\Hmax^\eps(A|B)_\rho \leq \max_{i} \Hmax^\eps(A|B)_{\rho^i} + 2 \log N.
	\end{align*}
\end{lemma}

\Cref{lem:min entropy mixture,lem:max entropy mixture} are proved in \cref{app:smooth entropies}. The smooth conditional min- and max-entropy satisfy the following asymptotic equipartition property~\cite{tomamichel2009fully,tomamichel2015quantum}:
\begin{align}
\label{eq:AEP H_min}
  \frac{1}{n}\Hmin^\eps(A^n|B^n)_{\rho^{\otimes n}} \geq H(A|B)_\rho - \frac {\delta(\eps,\rho)} {\sqrt n}, \\
\label{eq:AEP H_max}
  \frac{1}{n}\Hmax^\eps(A^n|B^n)_{\rho^{\otimes n}} \leq H(A|B)_\rho + \frac {\delta(\eps,\rho)} {\sqrt n}.
\end{align}
where $\rho_{AB} \in \mathcal S(AB)$ and $n\in\mathbb{N}$. For $n \geq \frac 8 5 \log \frac 2 {\eps^2}$ the convergence parameter can be bounded as:
\begin{align}\label{eq:AEP convergence}
  \delta(\eps, \rho)
  \leq 4 \log \bigl( 2 \sqrt{d_A} + 1 \bigr) \sqrt{\log \frac 2 {\eps^2}}.
\end{align}
For technical reasons we will also work with the conditional collision entropy defined as~\cite{renner2005security},
\begin{align}\label{eq:renyi 2}
	H_2(A|B)_\rho := -\log \min_{\sigma \in \mathcal S(B)}\tr \left[\left(\sigma_B^{-1/4}\rho_{AB}\sigma_B^{-1/4}\right)^2\right],
\end{align}
where $\rho_{AB} \in \mathcal S_\leq(AB)$. We have $H_2(A|B)_\rho\geq\Hmin(A|B)_\rho.$

\section{Uninformed users and informed receiver}\label{sec:uninformed and IR}

Since the two capacities advertised in \cref{thm:IR,thm:uninformed} match, it suffices to prove a coding theorem for uninformed users together with a converse bound for the informed receiver scenario. We start in \cref{subsec:uninformed one-shot} by establishing a one-shot coding theorem for finite compound channels by a well-known reduction to a single average channel and properties of smooth entropies. In \cref{subsec:uninformed direct} we obtain the direct part of \cref{thm:uninformed} by using our one-shot result in the limit of many channel uses combined with discretization techniques. Lastly, in \cref{subsec:IR converse} we establish the converse part of \cref{thm:IR}.

\subsection{One-shot coding theorem}\label{subsec:uninformed one-shot}

To establish our one-shot coding result, we use an equivalence between finite compound channels and average channels, whose analogue for the plain quantum capacity has been observed previously in~\cite[Lemma II.1]{bjelakovic2008quantum}.\footnote{We must keep in mind that this equivalence is only known to hold true for weak capacities as defined in~\eqref{eq:capacity}, and not for strong converse capacities.} Let $\Pi_{A\to B} = \{ \mathcal N^i_{A\to B}\}_{i=1}^N$ be a finite compound channel and consider the average channel $\overline{\Pi}_{A\to B} := \frac1N \sum_{i=1}^N \mathcal N^i_{A\to B}$. If $\mathcal E$ and $\mathcal D$ denote encoder and decoder for $\overline{\Pi}_{A\to B}$ that achieve an entanglement fidelity of
\begin{align*}
&F(\Phi^+_{A_0R},\mathcal D \circ \overline\Pi \circ \mathcal E(\Phi^+_{A_0R}\otimes\Phi^+_{A_1 B_1})) \\
&=\frac1N \sum_{i=1}^N F(\Phi^+_{A_0R},\mathcal D \circ \mathcal N^i \circ \mathcal E(\Phi^+_{A_0R}\otimes\Phi^+_{A_1 B_1}))\\
&\geq 1-\delta
\end{align*}
then it is immediate that, for all $i=1,\dots,N$,
\begin{align}\label{eq:entanglement fidelity union bound}
  F(\Phi^+_{A_0R},\mathcal D \circ \mathcal N^i \circ \mathcal E(\Phi^+_{A_0R}\otimes\Phi^+_{A_1 B_1})) \geq 1-N\delta.
\end{align}
Thus, the entanglement infidelity does increase by no more than a constant factor $N$ when we apply a code to the compound channel instead of the average channel. Conversely, any code for the compound channel achieves the same entanglement fidelity for the average channel. 


We now prove the following one-shot coding theorem for finite compound channels:

\begin{theorem}\label{thm:uninformed one-shot}
  For any finite compound channel $\Pi_{A\to B} = \{\mathcal N^i_{A\to B}\}_{i=1}^N$, pure state $\rho_{AA'}$, integers $M_0$ and $M_1$, and $\eps > 0$, there exist quantum operations $\mathcal E_{A_0A_1\to A}$ and $\mathcal D_{BB_1\to A_0}$, where $d_{A_0} = M_0$ and $d_{A_1} = d_{B_1} = M_1$, such that
  \begin{align*}
  &\min_i F(\Phi^+_{A_0R},\mathcal D \circ \mathcal N^i \circ \mathcal E(\Phi^+_{A_0R}\otimes\Phi^+_{A_1 B_1}))  \\
  \geq\;&1 - 4 N \sqrt{2 \sqrt{\delta_1} + \delta_2},
  \end{align*}
  where
  \begin{align*}
    \delta_1 &= 3 \cdot 2^{-\frac12 (\Hmin^\eps(A)_\rho - \log M_0 - \log M_1)} + 24 \eps, \\
    \delta_2 &= 3 \cdot 2^{-\frac12 (-\max_i \Hmax^\eps(A'|B)_{\mathcal N^i(\rho)} - 2 \log N - \log M_0 + \log M_1)}\\
    &\quad + 24 \eps.
  \end{align*}
\end{theorem}

Here, $\Hmin^\eps(A)_\rho$ stands for the smooth min-entropy, defined as in~\eqref{eq:smooth min entropy} with $B$ a one-dimensional system.

\begin{proof}
  We apply the one-shot coding theorem~\cite[Theorem 3.14, cf.~Theorem 3.15]{dupuis2009decoupling} to the average channel $\overline\Pi_{A\to B}$ to obtain quantum operations $\mathcal E$ and $\mathcal D$ such that
  \begin{align*}
    \lVert \mathcal D \circ \overline\Pi \circ \mathcal E (\Phi^+_{A_0R} \otimes \Phi^+_{A_1B_1}) - \Phi^+_{A_0R} \rVert_1 \leq 2 \sqrt{2 \sqrt{\delta_1} + \delta_2}
  \end{align*}
  where
  \begin{align*}
    \delta_1 &= 3 \cdot 2^{-\frac12 (\Hmin^\eps(A)_\rho - \log M_0 - \log M_1)} + 24 \eps, \\
    \delta_2 &= 3 \cdot 2^{-\frac12 (-\Hmax^\eps(A'|B)_{\overline\Pi(\rho)} - \log M_0 + \log M_1)} + 24 \eps.
  \end{align*}
  Using the relation between the fidelity and the trace distance in~\eqref{eq:fidelity lower bound one is normalized}, we obtain
  \begin{align*}
    F(\Phi^+_{A_0R},\mathcal D \circ \overline\Pi \circ \mathcal E(\Phi^+_{A_0R}\otimes\Phi^+_{A_1B_1}))
    \geq 1 - 4 \sqrt{2 \sqrt{\delta_1} + \delta_2}.
  \end{align*}
  and thus, from~\eqref{eq:entanglement fidelity union bound},
  \begin{align*}
    &F(\Phi^+_{A_0R},\mathcal D \circ N^i \circ \mathcal E(\Phi^+_{A_0R}\otimes\Phi^+_{A_1B_1}))  \\
    &\geq 1 - 4 N \sqrt{2 \sqrt{\delta_1} + \delta_2}
  \end{align*}
  for all $i=1,\dots,N$.
  On the other hand, we have that $\overline\Pi(\rho) = \frac1N \sum_{i=1}^N \mathcal N_i(\rho)$.
  Thus, we can apply \cref{lem:max entropy mixture} in \cref{app:smooth entropies}, which asserts that
  \begin{align*}
  \Hmax^\eps(A'|B)_{\overline\Pi(\rho)} \leq \max_i \Hmax^\eps(A'|B)_{\mathcal N_i(\rho)} + 2 \log N,
  \end{align*}
  and conclude that
  \begin{align*}
  \delta_2\;&\leq 3 \cdot 2^{-\frac12 (-\max_i \Hmax^\eps(A'|B)_{\mathcal N_i(\rho)} - 2 \log N - \log M_0 + \log M_1)} \\
  &\quad + 24 \eps. \qedhere
  \end{align*}
\end{proof}

\subsection{Achievability}\label{subsec:uninformed direct}

We now establish the direct part of \cref{thm:uninformed}:

\begin{lemma}\label{lem:uninformed direct}
Let $\Pi_{A \rightarrow B}  = \{ \mathcal N^i_{A\rightarrow B} \}_{i \in I}$ be a compound channel with arbitrary index set $I$. Then, we have
\begin{align*}
Q_E(\Pi_{A \rightarrow B}) \geq \frac12 \sup_{\rho} \inf_{i\in I}I(A':B)_{\mathcal{N}^i(\rho)},
\end{align*}
where the supremum is over all purified input distributions $\rho_{A'A}$ with $A \cong A'$.
\end{lemma}

\begin{proof}
  Let $\rho_{AA'}$ be a pure state, $\Delta > 0$, and
  \begin{align*}
  R = \frac12 \inf_{i \in I} I(A':B)_{\mathcal N^i(\rho)} - \Delta.
  \end{align*}
  We will show that for any $\delta > 0$ there exists $n_0$ such that, for all $n \geq n_0$, the triple $(R,n,\delta)$ is achievable for $\Pi_{A\to B}$.
  If $R \leq 0$ then there is nothing to show, thus we may assume that $R > 0$.

  To reduce to finite compound channels, we use the discretization result~\cite[Lemma V.2]{bjelakovic2008quantum}, which asserts that for any $\nu > 0$ there exists a finite compound channel $\widetilde\Pi_{A\to B} = \{ \widetilde{\mathcal N}^j_{A\to B} \}_{j=1}^N$ of cardinality
  $N \leq (6 / \nu)^{2 d_{AB}^2}$
  with the property that for any $\mathcal N^i \in \Pi$ there exists some $\widetilde{\mathcal N}^j \in \widetilde\Pi$ such that
  $\lVert (\mathcal N^i)^{\otimes k} - (\widetilde{\mathcal N}^j)^{\otimes k} \rVert_\diamond \leq k \nu$ for all $k$,
  and vice versa (that is,  for any $\widetilde{\mathcal N}^j \in \widetilde\Pi$ there exists some $\mathcal N^i \in \Pi$ such that $\lVert (\mathcal N^i)^{\otimes k} - (\widetilde{\mathcal N}^j)^{\otimes k} \rVert_\diamond \leq k \nu$ for all $k$).
  We shall choose $\nu = 1/n^2$.
  Then, the discretization has cardinality $N \leq (3 n)^{4 d_{AB}^2}$.
  Moreover, $\lVert \mathcal N^i - \widetilde{\mathcal N}^j \rVert_\diamond \leq 1 / {n^2}$ for any pair of channels as above, and thus~\eqref{eq:fannes mutual info} implies that
  \begin{align*}
  &\bigl| \inf_{i \in I} I(A':B)_{\mathcal N^i(\rho)} - \min_{j=1,\dots,N} I(A':B)_{\widetilde{\mathcal N}^j(\rho)} \bigr|  \\
  \leq\;&\frac 1 {n^2} \log d_A + \left(1+\frac{1}{2n^2}\right)h\biggl( \frac 1 {2n^2+1}\biggr).
  \end{align*}
  As a consequence,
  \begin{align}\label{eq:R vs discretization}
    R \leq \frac12 \min_{j=1,\dots,N} I(A':B)_{\widetilde{\mathcal N}^j(\rho)} - \frac\Delta2
  \end{align}
  for $n$ sufficiently large (depending only on $d_A$ and $\Delta$). Let us assume that this is the case.

  We now apply our one-shot coding result, \cref{thm:uninformed one-shot}, to $\widetilde\Pi^{\otimes n}_{A\to B} = \{ (\widetilde{\mathcal N}^j_{A\to B})^{\otimes n} \}_{j=1}^N$,  $\rho_{AA'}^{\otimes n}$, $M_0 = \lceil 2^{nR} \rceil$, $M_1 = \lceil 2^{n(H(A)_\rho - R - \Delta/2)} \rceil$.
  For all $\eps > 0$, we obtain an encoder $\mathcal E_{A_0A_1\to A^n}$ and a decoder $\mathcal D_{B^nB_1\to A_0}$ such that
  \begin{equation}\label{eq:uninformed F discretized}
  \begin{aligned}
    &\min_{j=1,\dots,N} F(\Phi^+_{A_0R},\mathcal D \circ (\widetilde{\mathcal N}^j)^{\otimes n} \circ \mathcal E(\Phi^+_{A_0R}\otimes\Phi^+_{A_1 B_1}))   \\
    &\geq 1 - 4 N \sqrt{2 \sqrt{\delta_1} + \delta_2},
  \end{aligned}
  \end{equation}
  where
 \begin{align*}
\delta_1\;&= 3 \cdot 2^{-\frac12 (\Hmin^\eps(A^n)_{\rho^{\otimes n}} - \log M_0 - \log M_1)} + 24 \eps, \\
\delta_2\;&= 3 \cdot 2^{-\frac12 (-\max_{j} \Hmax^\eps({A'}^n|B^n)_{(\widetilde{\mathcal N}^j(\rho))^{\otimes n}} - 2 \log N)} \\
&\qquad\cdot 2^{-\frac12 (- \log M_0 + \log M_1)} + 24 \eps.
  \end{align*}
  We now choose\footnote{We make this choice for convenience and note that the use of the asymptotic equipartition property~\eqref{eq:AEP H_min} would also allow to choose $\varepsilon$ super-polynomially small in $n$.} $\eps = 1/(n N)^4$. Since $\eps$ decays only polynomially with $n$, the asymptotic equipartition property \eqref{eq:AEP H_min} together with~\eqref{eq:AEP convergence} and the estimate~\eqref{eq:R vs discretization} implies that
  \begin{align*}
    \delta_1
    \leq 3 \cdot 2^{-\frac n 2 \bigl( \frac \Delta 2 - \frac 2 n - \frac {\delta(\eps,\rho)} {\sqrt n} \bigr)} + 24 \eps
    \leq 3 \cdot 2^{-n \frac \Delta 8} + 24 \eps
    \leq 25 \eps
  \end{align*}
  for sufficiently large $n$ (depending only $d_{AB}$, $\Delta$ and $R$).
  Likewise, using \eqref{eq:AEP H_max} instead of \eqref{eq:AEP H_min} we obtain that
  \begin{align*}
    \delta_2\;&\leq 3N \cdot 2^{-\frac n 2 \bigl(\frac \Delta 2 - \frac 1 n - \frac {\delta(\eps,\widetilde{\mathcal N}^j(\rho))} {\sqrt n} \bigr)} + 24 \eps  \\
    &\leq 3  N \cdot 2^{-n \frac \Delta 8} + 24 \eps\\
    &\leq 25 \eps,
  \end{align*}
  where we use that $N$ grows only polynomially with $n$.
  By inserting the two bounds into~\eqref{eq:uninformed F discretized}, we obtain that
  \begin{align*}
    &\quad F(\Phi^+_{A_0R},\mathcal D \circ (\widetilde{\mathcal N}^j)^{\otimes n} \circ \mathcal E(\Phi^+_{A_0R}\otimes\Phi^+_{A_1 B_1}))  \\
    &\geq  1 - 4 N \sqrt{2 \sqrt{25 \eps} + 25 \eps}
    \geq 1 - 24 N \eps^{1/4}\\
    &= 1 - \frac {24} n.
  \end{align*}
  At last, we relate this to the entanglement fidelity for the original compound channel. Using~\eqref{eq:fidelity channel continuity} and recalling that $\nu = 1/n^2$, we find that
  \begin{align*}
    &\quad \inf_{i \in I} F(\Phi^+_{A_0R},\mathcal D \circ (\mathcal N^i)^{\otimes n} \circ \mathcal E(\Phi^+_{A_0R}\otimes\Phi^+_{A_1 B_1}))  \\
    &\geq  1 - \frac {24} n - \frac 2 {\sqrt n},
  \end{align*}
  since for any $\mathcal N^i \in \Pi$ there exist $\widetilde{\mathcal N}^j \in \widetilde\Pi$ such that $\lVert (\mathcal N^i)^{\otimes n} - (\widetilde{\mathcal N}^j)^{\otimes n} \rVert_\diamond \leq n\nu = 1/n$. We conclude that, for any $\delta > 0$ and sufficiently large $n$, $(R,\delta,n)$ is a valid triple for the compound channel $\Pi_{A\to B}$.
\end{proof}

\subsection{Converse}\label{subsec:IR converse}

The following lemma establishes the converse direction of \cref{thm:IR}. The proof uses ideas from~\cite{adami97,matthews2014finite,Datta2008}.

\begin{lemma}\label{lem:IR converse}
	Let $\Pi_{A \rightarrow B}  = \{ \mathcal N^i_{A\rightarrow B} \}_{i \in I}$ be a compound channel with arbitrary index set $I$.
	Then, we have
	\begin{align*}
		Q_{E,IR}(\Pi_{A \rightarrow B}) \leq \frac12 \sup_{\rho} \inf_{i\in I}I(A':B)_{\mathcal{N}^i(\rho)},
	\end{align*}
	where the supremum is over all purified input distributions $\rho_{A'A}$ with $A \cong A'$.
\end{lemma}

\begin{proof}
  In view of~\eqref{eq:classical vs quantum}, it suffices to argue that
  \begin{align}\label{eq:classical todo}
  	C_{E,IR}(\Pi_{A\to B}) \leq \sup_\rho \inf_{i \in I} I(A':B)_{\mathcal N^i(\rho)},
  \end{align}
  where $C_{E,IR}$ stands for the entanglement-assisted classical capacity with informed receiver as discussed in \cref{sec:results}.
  Thus, let $(R,n,\delta)$ be an achievable triple for entanglement-assisted classical communication in the informed receiver scenario, with corresponding code $(M_0,M_1,\{\mathcal E^m\}, \{ \Lambda^{m,i} \}_{m,i})$, where $\{ \Lambda^{m,i} \}_{m=1}^{M_0}$ is a POVM for each fixed $i \in I$. 
  Let $\rho_{A^n} = \frac 1 {M_0} \sum_{m=1}^{M_0} \mathcal E^m(\tau_{M_1})$ denote the average channel input.
  For each fixed $i \in I$, $(M_0,M_1,\{\mathcal E^m\}, \{ \Lambda^{m,i} \}_m)$ is a code for entanglement-assisted classical communication through $(\mathcal N^i)^{\otimes n}$ with message size $M_0$ and error probability $\delta$. Thus, we may apply the converse from~\cite[Theorem 18, (43) \& Lemma 30]{matthews2014finite},
	\begin{align}\label{eq:matthews_wehner}
	\log M_0 \leq \frac 1 {1-\delta} \Big( I(\rho_{A^n}, (\mathcal N^i)^{\otimes n}) + h(\delta) \Big),
	\end{align}
  where we have introduced the notation $I(\sigma_A, \mathcal T) := I(A':B)_{\mathcal T[\sigma_{AA'}]}$, with $\sigma_{AA'}$ an arbitrary purification of $\sigma_A$.
  We now use the sub-additivity property~\cite[(3.24)]{adami97},
  \begin{align*}
  I(\rho_{A^n}, (\mathcal N^i)^{\otimes n}) \leq \sum_{k=1}^n I(\rho^k_A, \mathcal N^i),
  \end{align*}
  where we have defined $\rho^k_A := \rho_{A_k}$.
  By definition, the right-hand side is equal to
  \begin{align*}
  \sum_{k=1}^n I(\rho^k_A, \mathcal N^i) = \sum_{k=1}^n I(A':B)_{\mathcal N^i(\rho^k_{AA'})},
  \end{align*}
  where the $\ket{\rho^k_{AA'}}$ denote purifications of the $\rho^k_A$. To further upper bound this expression, we introduce the pure state
  \begin{align}\label{eq:indep of channel}
  	\sigma_{AA'XY} := \frac 1 {\sqrt n} \sum_{k=1}^n \ket{\rho_{AA'}^k} \otimes \ket{kk}_{XY}.
  \end{align}
  Then, $\sigma_{AA'X} = \frac 1 n \sum_{k=1}^n \rho_{AA'}^k \otimes \ket k\!\!\bra k_X$, and therefore
  \begin{equation}\label{eq:chain rule stuff}
  \begin{aligned}
  	&\quad\frac 1 n \sum_{k=1}^n I(A':B)_{\mathcal N^i(\rho_{AA'}^k)}
  	= I(A':B|X)_{\mathcal N^i(\sigma)} \\
  	&= I(A'X:B)_{\mathcal N^i(\sigma)} - I(B:X)_{\mathcal N^i(\sigma)} \\
  	&\leq I(A'X:B)_{\mathcal N^i(\sigma)}
  	\leq I(A'XY:B)_{\mathcal N^i(\sigma)} \\
  	&= I(\sigma_A, \mathcal N^i),
  \end{aligned}
  \end{equation}
  where the first equation uses that $X$ is classical, the second is the chain rule for the conditional mutual information, the third step is the non-negativity of the mutual information, the fourth the monotonicity of mutual information under local quantum operations, and the last equation is again by definition. If we plug~\eqref{eq:chain rule stuff} into~\eqref{eq:matthews_wehner} then we obtain the upper bound
  \begin{align*}
  \frac 1 n \log M_0 \leq \frac 1 {1-\delta}\left(I(\sigma_A, \mathcal N^i)+\frac{h(\delta)}{n}\right).
  \end{align*}
  Crucially, the state $\sigma_A = \frac 1 n \sum_{k=1}^n \rho_{A}^k$ does not depend on the channel $\mathcal N^i$ under consideration (see~\eqref{eq:indep of channel}). Thus, it follows that the above inequality holds for all $i \in I$,
  \begin{align*}
  \frac 1 n \log M_0 \leq \frac 1 {1-\delta}\left(\inf_{i \in I} I(\sigma_A, \mathcal N^i)+\frac{h(\delta)}{n}\right),
  \end{align*}
  and therefore
  \begin{align*}
  R &\leq \frac 1 n \log M_0\\
  &\leq \frac 1 {1-\delta}\left(\sup_\rho \inf_{i \in I} I(\rho_A, \mathcal N^i)+\frac{h(\delta)}{n}\right) \\
  &=  \frac 1 {1-\delta}\left(\sup_\rho \inf_{i \in I} I(A':B)_{\mathcal N^i(\rho)}+\frac{h(\delta)}{n}\right).
  \end{align*}
  This establishes~\eqref{eq:classical todo} and thus the claim of the theorem.
\end{proof}

As explained at the beginning of \cref{sec:uninformed and IR}, \cref{lem:uninformed direct,lem:IR converse} together establish \cref{thm:uninformed,thm:IR}.

\section{Informed sender}\label{sec:informed_sender}

\subsection{One-shot coding theorem}

Our one-shot coding theorem in the uninformed scenario can be understood a direct consequence of a corresponding result in~\cite{dupuis2009decoupling} for a single fixed channel, applied to the average channel induced by the compound. In the informed sender scenario, such a reduction is complicated by the fact that now the encoders depend on the individual channels in the compound. In the case of the plain quantum capacity, we show how these challenges can in fact be overcome by a suitable reduction, which leads to a pleasant new proof of the corresponding result in~\cite{bjelakovic09} (see \cref{app:plain IS}). In the presence of entanglement assistance, however, we need to develop some new technical tools.

Following the decoupling approach, we start with the following ansatz for the encoders~\cite{dupuis2009decoupling}: Given integers $M_0$ and $\{M_1^i\}_{i=1}^N$, let $M_1$ denote the least common multiple of the $M_1^i$. Let $A_0$ and $A_1$ denote quantum systems of dimensions $M_0$ and $M_1$, respectively, and fix for each value of $i$ a tensor product decomposition $A_1 \cong A_1^i \otimes (A_1^i)^c$ such that $d_{A_1^i} = M_1^i$. Given states $\{\rho^i_A\}_{i=1}^N$, we now define completely positive maps
\begin{equation}\label{eq:IS ansatz}
\begin{aligned}
  &\mathcal E^i_{A_0A_1\to A}(\sigma_{A_0A_1}) := d_A O_A(\rho^i) U^i_A J^i_{A_0A_1^i\to A}  \\
  &\quad\sigma_{A_0A_1^i} (J^i_{A_0A_1^i\to A})^\dagger (U^i_A)^\dagger \bigl(O_A(\rho^i)\bigr)^\dagger,
\end{aligned}
\end{equation}
where the $J^i_{A_0A^i_1\to A}$ are fixed full-rank partial isometries, the $U^i_A$ denote unitaries that will later be chosen at random, and where we use the notation
\begin{equation}\label{eq:OA definition}
  O_A(\rho) := \sum_{a,a'} \rho_{a,a'} \ket a\!\!\bra{a'}_A
\end{equation}
with $\rho_{a,a'}$ the coefficients obtained by expanding the pure state $\ket{\rho_{AA'}}$ in the same computational basis as our maximally entangled states, i.e., $\ket{\rho_{AA'}}= \sum_{a,a'} \rho_{a,a'} \ket{a_A}\ket{a'_{A'}}$.
Then we have $\ket{\rho_{AA'}}=\sqrt{d_A} O_A(\rho) \ket{\Phi^+_{AA'}}$. We caution that $O_A(\rho)$ is not in general Hermitian. 

To assess the performance of the encoders $\{\mathcal E^i\}_{i=1}^N$, we will consider the average encoder-and-channel
\begin{align}\label{eq:average motivation}
\overline{\mathcal T}_{A_0A_1\to B} :=  \frac1N \sum_{i=1}^N \mathcal N^i_{A\to B} \circ \mathcal E^i_{A_0A_1\to A}
\end{align}
and show that the complementary map decouples the reference $R$ from the environment. Following the decoupling approach, this will guarantee the existence of an uninformed decoder $\mathcal D_{B_1B\to A}$ for the map $\overline{\mathcal T}_{A_0A_1\to B}$ and therefore, as in \cref{subsec:uninformed one-shot} above, for each of its branches $\mathcal N^i \circ \mathcal E^i$.

For the purposes of obtaining a decoupling result in terms of smooth entropies it will in fact be useful to consider more general maps of the form
\begin{equation}\label{eq:average general}
\begin{aligned}
&\overline{\mathcal T}_{A_0A_1\to B}(\sigma_{A_0A_1}) := \frac1N \sum_{i=1}^N \mathcal T^i_{A\to B} \bigl( U^i_A J^i_{A_0A_1^i\to A} \\
&\quad\sigma_{A_0A_1^i} (J^i_{A_0A_1^i\to A})^\dagger (U^i_A)^\dagger \bigr),
\end{aligned}
\end{equation}
where the $\mathcal T^i_{A\to B}$ are arbitrary completely positive maps; we recover~\eqref{eq:average motivation} for the choice
\begin{align*}
\mathcal T^i_{A\to B}(\sigma_A) = \mathcal N^i_{A\to B}\left(d_A O_A(\rho^i) \sigma_A \left( O_A(\rho^i) \right)^\dagger\right).
\end{align*}
We now obtain an explicit complementary map. For this, let $K^i_{(A_1^i)^c\to A^c}$ denote isometries, where $A^c$ is an auxiliary system of sufficiently large dimension (e.g., $M_1$), and let $W^i_{A\to BE}$ denote dilations of the maps $\mathcal T^i_{A\to B}$.
Then, the maps
\begin{align}\label{eq:single complement}
\mathcal T^{i,c}_{A\to E}(\sigma_A) := \tr_B\left[W^i_{A\to BE} \sigma_A (W^i_{A\to BE})^\dagger\right]
\end{align}
define complementary maps of the channels $\mathcal T^i_{A\to B}(\sigma_A) = \tr_E\left[W^i_{A\to BE} \sigma_A (W^i_{A\to BE})^\dagger\right]$, and it is not hard to verify that the completely positive map
\begin{equation}\label{eq:average general complement}
\begin{aligned}
& \overline{\mathcal T}^c_{A_0A_1\to A^cEI}(\sigma_{A_0A_1}):=\frac1N \sum_{i,j}\tr_B \big[ W^i_{A\to BE} U_A^i \\
&\quad J^i_{A_0A_1^i\to A} K^i_{(A_1^i)^c\to A^c} \sigma_{A_0A_1} (K^j_{(A_1^j)^c\to A^c})^\dagger \\
&\quad (J^j_{A_0A_1^j\to A})^\dagger (U_A^j)^\dagger (W^j_{A\to BE})^\dagger\big] \otimes \ket i\!\!\bra j_I
\end{aligned}
\end{equation}
is complementary to the map~\eqref{eq:average general}.

\begin{lemma}\label{lem:average decoupling non-smooth}
Let $\{\mathcal T^i_{A\to B}\}_{i=1}^N$ be completely positive maps, and $M_0$, $\{M_1^i\}_{i=1}^N$ integers such that $M_0 M_1^i \leq d_A$ for all $i$. Let $\{ \mathcal T^{i,c}_{A\to E} \}_{i=1}^N$ and $\overline{\mathcal T}^c_{A_0A_1\to A^cEI}$ denote the complementary maps as defined in~\eqref{eq:single complement} and~\eqref{eq:average general complement}, respectively. Then, we have
\begin{align*}
&\quad\EE \left\lVert \overline{\mathcal T}^c_{A_0A_1\to A^cEI}(\Phi^+_{A_0R} \otimes \tau_{A_1}) - \omega_{A^cEI} \otimes \tau_R \right\rVert_1  \\
&\leq 2^{-\frac12 \bigl( \min_i H_2(A'|E)_{\mathcal T^{i,c}(\Phi^+)} -\log M_0 + \log M_1^i - 2 \log N - 2 \bigr)},
\end{align*}
where $\EE$ denotes the average over independent Haar-random unitaries $\{U^i_A\}$, and
\begin{align*}
\omega_{A^cEI} := \frac1N \sum_{i=1}^N \, &K^i_{(A_1^i)^c\to A^c} \, \tau_{(A_1^i)^c} (K^i_{(A_1^i)^c\to A^c})^\dagger \\
\otimes\;&\mathcal T^{i,c}_{A\to E} (\tau_A) \otimes \ket i\!\!\bra i_I.
\end{align*}
\end{lemma}
\begin{proof}
  We start by bounding the trace norm deviation from the average state by using the triangle inequality:
\begin{align}
\label{eq:deviation to bound}
    &\quad \EE \left\lVert \overline{\mathcal T}^c_{A_0A_1\to A^cEI}(\Phi^+_{A_0R} \otimes \tau_{A_1}) - \omega_{A^c E I} \otimes \tau_R \right\rVert_1 \\
\nonumber
    &\leq \frac1N \sum_{i,j} \EE \Big\lVert \Bigl(
			\tr_B\Big[( W^i U^i J^i K^i (\Phi^+_{A_0R} \otimes \tau_{A_1}) \\
\nonumber
			&\qquad\qquad\qquad\qquad \cdot (K^j)^\dagger (J^j)^\dagger (U^j)^\dagger (W^j)^\dagger\Big] \\
\nonumber
			&\quad - \delta_{i,j} K^i \, \tau_{(A_1^i)^c} (K^i)^\dagger \otimes \mathcal T^{i,c} (\tau_A) \otimes \tau_R \Bigr) \otimes \ket i\!\!\bra j_I \Big\rVert_1 \\
  \label{eq:diagonal}
    &\leq \frac1N \sum_{i=1}^N \EE \Big\lVert \mathcal T^{i,c} \bigl( U^i J^i (\Phi^+_{A_0R} \otimes \tau_{A_1^i}) (J^i)^\dagger (U^i)^\dagger \bigr) \\
\nonumber
     &\qquad\qquad\quad  - \mathcal T^{i,c}(\tau_A^i) \otimes \tau_R \Big\rVert_1 \\
  \label{eq:offdiagonal}
    &\quad + \frac1N \sum_{i\neq j} \EE \Big\lVert \tr_B\Big[W^i U^i J^i K^i (\Phi^+_{A_0R} \otimes \tau_{A_1}) \\
\nonumber
    &\qquad\qquad\qquad \cdot (K^j)^\dagger (J^j)^\dagger (U^j)^\dagger (W^j)^\dagger\Big] \Big\rVert_1.
 \end{align}
 To bound the averages in \eqref{eq:diagonal}, we invoke the one-shot decoupling theorem~\cite[Theorem 3.3]{dupuis2014one} to obtain the first inequality in
  \begin{align}
  \nonumber
  	&\EE \left\lVert \mathcal T^{i,c} \bigl( U^i J^i (\Phi^+_{A_0R} \otimes \tau_{A_1^i}) (J^i)^\dagger (U^i)^\dagger \bigr)
      - \mathcal T^{i,c}(\tau_A^i) \otimes \tau_R \right\rVert_1 \\
  \nonumber
    & \leq 2^{-\frac12 \bigl( H_2(A'|E)_{\mathcal T^{i,c}(\Phi^+)} - \log M_0 + \log M_1^i \bigr)}\\
  \nonumber
    &\leq \sqrt{\tr\left[\left({\widetilde\rho}^i_{A'E}\right)^2\right]} \cdot 2^{-\frac12( - \log M_0 + \log M_1^i )}\\
    &= \sqrt{\tr\left[\left({\widetilde\rho}^i_{B}\right)^2\right]} \cdot 2^{-\frac12( - \log M_0 + \log M_1^i )}
    =: x_{ii}.\label{eq:diagonal bound}
  \end{align}
  For the second inequality we have defined $\rho_{A'BE}^i := W_{A\to BE}^i \Phi^+_{AA'} (W_{A\to BE}^i)^\dagger$ and ${\widetilde\rho}_{A'BE}^i := (\alpha_E^i)^{-1/4} \rho_{A'BE}^i (\alpha_E^i)^{-1/4}$ for an arbitrary choice of state $\alpha_E^i \in \mathcal S(E)$, and used the definition of the conditional collision entropy in~\eqref{eq:renyi 2}.

  Bounding the averages in \eqref{eq:offdiagonal} is somewhat more involved because we cannot directly rely on previous results.
  We start with the H\"older inequality~\cite[Corollary IV.2.6]{bhatia2013matrix} and obtain the upper bound $\EE \lVert \kappa^{ij}_{A^cER} \rVert_2$, with $\kappa^{ij}_{A^cBER}$ defined as in~\eqref{eq:hoelder};
  \begin{figure*}[!ht]
\begin{equation}\label{eq:hoelder}
  \begin{aligned}
  	&\quad \EE \left\lVert \tr_B\left[W^i U^i J^i K^i (\Phi^+_{A_0R} \otimes \tau_{A_1}) (K^j)^\dagger (J^j)^\dagger (U^j)^\dagger (W^j)^\dagger\right]\right\rVert_1 \\
    &\leq \bigl\lVert \left(\alpha^i_E \otimes I_R \otimes K^i_{(A_1^i)^c\to A^c} (K^i_{(A_1^i)^c\to A^c})^\dagger\right)^{1/4} \bigr\rVert_4
    \cdot \bigl\lVert \left(\alpha^j_E \otimes I_R \otimes K^j_{(A_1^j)^c\to A^c} (K^j_{(A_1^j)^c\to A^c})^\dagger\right)^{1/4} \bigr\rVert_4 \\
    &\quad\cdot \EE \left\lVert (\alpha^i_E)^{-1/4} \tr_B\left[W^i U^i J^i K^i (\Phi^+_{A_0R} \otimes \tau_{A_1}) (K^j)^\dagger (J^j)^\dagger (U^j)^\dagger (W^j)^\dagger\right] (\alpha^j_E)^{-1/4} \right\rVert_2 \\
    &= 2^{-\frac12 \bigl( -\log M_0 + \frac12 \log M_1^i + \frac12 \log M_1^j - \log M_1 \bigr)} \\
    &\quad\cdot \EE \Big\lVert \tr_B\Big[\underbrace{( \alpha^i_E )^{-1/4} W^i U^i J^i K^i (\Phi^+_{A_0R} \otimes \tau_{A_1}) (K^j)^\dagger (J^j)^\dagger (U^j)^\dagger (W^j)^\dagger ( \alpha^j_E )^{-1/4}}_{=: \kappa^{ij}_{A^cBER}}\Big]\Big\rVert_2,
  \end{aligned}
\end{equation}
\hrulefill
\end{figure*}
here, we have also used that the $K^i_{(A_1^i)^c\to A^c} (K^i_{(A_1^i)^c\to A^c})^\dagger$ are orthogonal projections onto the ranges of the isometries $K^i_{(A_1^i)^c\to A^c}$.
  We now apply Jensen's inequality and the swap trick,
  \begin{align}\label{eq:jensen swap}
  \begin{aligned}
  	&\quad\EE \lVert \kappa^{ij}_{A^cER} \rVert_2 \leq \sqrt{\EE \tr\left[\kappa^{ij}_{A^cER} (\kappa^{ij}_{A^cER})^\dagger\right]}  \\
	&=\sqrt{\EE \tr\left[\left(\kappa^{ij}_{A^cBER} \otimes (\kappa^{ij}_{A^cBER})^\dagger\right) F_{A^cER}\right]},
  \end{aligned}
  \end{align}
  where we write $F_S$ for the operator that swaps two copies of a subsystem $S$ and acts as the identity otherwise.
  To compute the right-hand side average, it will be useful to introduce the following notation:
{\small \begin{align*}
&{\widetilde W}_{A\to BE}^i := (\alpha^i_E)^{-1/4} W^i_{A\to BE} \\
&\Omega_{AA^cR}^{ij} := J^i K^i (\Phi^+_{A_0R} \otimes \tau_{A_1}) (K^j)^\dagger (J^j)^\dagger,
\end{align*}}
so that
\begin{align*}
\kappa^{ij}_{A^cBER} &= {\widetilde W}_{A\to BE}^i U_A^i \Omega_{AA^cR}^{ij} (U_A^j)^\dagger ({\widetilde W}_{A\to BE}^j)^\dagger \\
&= (\kappa^{ji}_{A^cBER})^\dagger.
\end{align*}
  Then, we get
  \begin{align*}
		&\quad \EE \tr\left[(\kappa^{ij}_{A^cBER} \otimes (\kappa^{ij}_{A^cBER})^\dagger) F_{A^cER}\right] \\
		&= \EE \tr\Big[\bigl( {\widetilde W}_{A\to BE}^i U_A^i \Omega_{AA^cR}^{ij} (U_A^j)^\dagger ({\widetilde W}_{A\to BE}^j)^\dagger \\
		&\qquad  \otimes {\widetilde W}_{A\to BE}^j U_A^j \Omega_{AA^cR}^{ji} (U_A^i)^\dagger ({\widetilde W}_{A\to BE}^i)^\dagger \bigr) F_{A^cER}\Big] \\
		&= d_A^{-2} \sum_{a,b,c,d} \tr\Big[\bigl({\widetilde W}_{A\to BE}^i \ket a\!\!\bra b \Omega_{AA^cR}^{ij} \ket d\!\!\bra c ({\widetilde W}_{A\to BE}^j)^\dagger \\
		&\qquad \otimes {\widetilde W}_{A\to BE}^j \ket c\!\!\bra d \Omega_{AA^cR}^{ji} \ket b\!\!\bra a ({\widetilde W}_{A\to BE}^i)^\dagger \bigr) \, F_{A^cER}\Big],
	\intertext{since $\EE(U_A^i \otimes (U_A^i)^\dagger) = d_A^{-1} F_A = d_A^{-1} \sum_{a,b} \ket a\!\!\bra b \otimes \ket b\!\!\bra a$ and likewise for $U_A^j$. This in turn is equal to}
		&\quad d_A^{-2} \sum_{a,c} \tr\Big[\bigl({\widetilde W}_{A\to BE}^i \ket a\!\!\bra c ({\widetilde W}_{A\to BE}^j)^\dagger \\
		&\qquad\qquad\quad\otimes {\widetilde W}_{A\to BE}^j \ket c\!\!\bra a ({\widetilde W}_{A\to BE}^i)^\dagger \bigr) F_E \Big]\\
		&\qquad\cdot \sum_{b,d} \tr\left[\bigl( \bra b \Omega_{AA^cR}^{ij} \ket d \otimes \bra d \Omega_{AA^cR}^{ji} \ket b \bigr) F_{A^cR}\right]\\
		&= d_A^{-2} \tr\left[ \bigl( {\widetilde W}_{A\to BE}^i ({\widetilde W}_{A\to BE}^i)^\dagger \otimes {\widetilde W}_{A\to BE}^j ({\widetilde W}_{A\to BE}^j)^\dagger \bigr) F_B \right] \\
		&\quad\cdot \tr \left[ \Omega_{AA^cR}^{ij} \Omega_{AA^cR}^{ji} \right]\\
		&= \tr\left[\bigl( {\widetilde\rho}_{BE}^i \otimes {\widetilde\rho}_{BE}^j \bigr) F_B\right]
		\cdot \tr \left[(\Phi^+_{A_0R} \otimes \tau_{A_1})^2\right]\\
		&= \tr \left[{\widetilde\rho}_B^i {\widetilde\rho}_B^j\right] \cdot 2^{-\log M_1}.
  \end{align*}
  If we combine this result with inequalities~\eqref{eq:hoelder} and \eqref{eq:jensen swap}, we obtain the following bound:
  \begin{align}
  	&\quad\EE \left\lVert \tr_B \left[ W^i U^i J^i K^i (\Phi^+_{A_0R} \otimes \tau_{A_1}) (W^j U^j J^j K^j)^\dagger\right]\right\rVert_1 \nonumber\\
  	&\leq \sqrt{\tr\left[{\widetilde\rho}_{B}^i {\widetilde\rho}_{B}^j\right]} \cdot
  	2^{-\frac12 ( -\log M_0 + \frac12 \log M_1^i + \frac12 \log M_1^j )} \nonumber \\
  	&=: x_{ij}.\label{eq:offdiagonal bound}
  \end{align}
  We thus obtain the following bound on~\eqref{eq:deviation to bound},
  \begin{align*}
 		&\quad \EE \lVert \overline{\mathcal T}^c_{A_0A_1\to A^cEI}(\Phi^+_{A_0R} \otimes \tau_{A_1}) - \omega_{A^c E I} \otimes \tau_R \rVert_1\\
 		&\leq \frac1N \sum_{i,j} x_{ij},
  \end{align*}
  where the $x_{ij}$ are defined in~\eqref{eq:diagonal bound} and~\eqref{eq:offdiagonal bound}.
  At last, we note as in~\cite[Lemma III.3]{bjelakovic2008quantum} that $x_{ij} \leq \sqrt{x_{ii} x_{jj}}$ by the Cauchy-Schwarz inequality, therefore $x_{ij} \leq \max(x_{ii}, x_{jj}) \leq x_{ii} + x_{jj}$, and that we can thus upper bound
  \begin{align*}
    &\quad \frac1N \sum_{i,j} x_{ij}\leq   \sum_{i=1}^N x_{ii} \\
    &=  2 \sum_{i=1}^N \sqrt{\tr\left[({\widetilde\rho}^i_{A'E})^2\right]} \cdot 2^{-\frac12( - \log M_0 + \log M_1^i )}  \\
    &\leq \max_i \sqrt{\tr\left[({\widetilde\rho}^i_{A'E})^2\right]} \cdot 2^{-\frac12( - \log M_0 + \log M_1^i - 2 \log N - 2 )}.
  \end{align*}
  This holds for all choices of $\alpha^i_E$ in the states ${\widetilde\rho}^i_{A'E} = (\alpha_E^i)^{-1/4} \mathcal N^{i,c}_{A\to E}(\Phi^+_{AA'}) (\alpha_E^i)^{-1/4}$. The claimed bound thus follows from the definition of the conditional collision entropy~\eqref{eq:renyi 2}.
\end{proof}

We now derive a smoothed version of \cref{lem:average decoupling non-smooth}. Later, this will allow us to treat the asymptotic IID limit using the asymptotic equipartition property in the form of~\eqref{eq:AEP H_min}-\eqref{eq:AEP H_max}. We note that this approach is conceptually different from previous works~\cite{bjelakovic2008quantum,bjelakovic09,bjelakovic2009entanglement}.

\begin{lemma}\label{lem:average decoupling smooth}
	Let $\Pi_{A\to B} = \{\mathcal N^i_{A\to B}\}_{i=1}^N$ be a finite compound channel, $\{\rho^i_{AA'}\}_{i=1}^N$ pure states, $\{\mathcal E^i_{A_0A_1\to A}\}_{i=1}^N$ the corresponding completely positive maps defined in~\eqref{eq:IS ansatz}, $M_0$, $\{M_1^i\}_{i=1}^N$ integers such that $M_0 M_1^i \leq d_A$ for all $i$, and $\eps > 0$.
	Then, there exists a quantum operation $\mathcal D_{BB_1\to A_0}$ which depends measurably on the random unitaries $\{U_A^i\}$ such that
	\begin{align*}
	&\EE \Big\lVert (\mathcal D_{BB_1\to A_0} \circ \frac1N \sum_{i=1}^N \mathcal N^i_{A\to B} \circ \mathcal E^i_{A_0A_1\to A})(\Phi^+_{A_0R} \otimes \Phi^+_{A_1B_1})  \\
	&\quad-\Phi^+_{A_0R} \Big\rVert_1 \leq \delta + 2\sqrt{2\delta} + 2 \eps,
	\end{align*}
	where
	\begin{align*}
	\delta = 2^{-\frac12 \bigl( -\max_i \Hmax^\eps(A'|B)_{\mathcal N^i(\rho^i)} - \log M_0 + \log M_1^i - 2 \log N - 2 \bigr)}.
	\end{align*}
\end{lemma}

\begin{proof}
  According to~\eqref{eq:smooth max entropy}, there exist $\widetilde\rho_{A'B}^i \in S_\leq(AB)$ such that
  \begin{equation}\label{eq:desmoothing}
  \begin{aligned}
  	&\Hmax^\eps(A'|B)_{\mathcal N^i(\rho^i)} = \Hmax(A'|B)_{\widetilde\rho^i} \quad\text{and} \\
	&P(\mathcal N^i_{A\to B}(\rho_{AA'}^i), \widetilde\rho_{A'B}^i) \leq \eps.
  \end{aligned}
  \end{equation}
  Let $\mathcal T^i_{A\to B}$ denote completely positive maps with $\mathcal T^i_{A\to B}(\Phi^+_{AA'}) = \widetilde\rho_{A'B}^i$ as their Choi-Jamiolkowski states, and define $\{ \mathcal T^{i,c}_{A\to E} \}_{i=1}^N$, $\overline{\mathcal T}_{A_0A_1\to A}$ and $\overline{\mathcal T}^c_{A\to A^cEI}$ as in~\eqref{eq:average general}--\eqref{eq:average general complement}.
  Then, we have
  \begin{equation}\label{eq:desmoothing complete}
  \begin{aligned}
    \Hmax^\eps(A'|B)_{\mathcal N^i(\rho^i)}
    &= \Hmax(A'|B)_{\mathcal T^i(\Phi^+)} \\
    &= -\Hmin(A'|E)_{\mathcal T^{i,c}(\Phi^+)}
  \end{aligned}
  \end{equation}
  by~\eqref{eq:smooth max entropy duality}. For each realization of the random unitaries $\{U_A^i\}$, \cref{lem:uhlmann one is normalized}
  implies that there exists a quantum operation $\mathcal D_{BB_1\to A_0}$ (partial isometry) such that
  \begin{align}\label{eq:uhlmann decoder bound}
  	&\left\lVert \mathcal D_{BB_1\to A_0} \circ \overline{\mathcal T}_{A_0A_1\to B}(\Phi^+_{A_0R} \otimes \Phi^+_{A_1B_1}) - \Phi^+_{A_0R} \right\rVert_1  \nonumber \\
  	&\leq \Delta + 2\sqrt{2\Delta},
  \end{align}
  where
  \begin{align*}
  \Delta := \left\lVert \overline{\mathcal T}_{A_0A_1\to A^cEI}^c(\Phi^+_{A_0R} \otimes \tau_{A_1}) - \omega_{A^cEI} \otimes \tau_R \right\rVert_1,
  \end{align*}
  with $\omega_{A^cEI}$ as defined in \cref{lem:average decoupling non-smooth}.
  In fact, $\mathcal D_{BB_1\to A_0}$ can be chosen as a measurable function of the $\{U_A^i\}$, so that it can itself be regarded as a random variable. Thus, it makes sense to bound the following expression:
  \begin{align}
  \label{eq:smoothing todo}
		&\EE \Big\lVert \bigl( \mathcal D \circ \frac1N \sum_{i=1}^N \mathcal N^i \circ \mathcal E^i \bigr)(\Phi^+_{A_0R} \otimes \Phi^+_{A_1B_1})
		- \Phi^+_{A_0R} \Big\rVert_1 \\
	\label{eq:smoothing first}
		&\leq \frac1N \sum_{i=1}^N \EE \Big\lVert \bigl( \widehat{\mathcal T}^i - \mathcal T^i \bigr)(U^i J^i (\Phi^+_{A_0R} \otimes \Phi^+_{A_1^iB_1^i}) \\
    &\qquad\qquad\qquad\qquad\qquad\qquad (J^i)^\dagger (U_A^i)^\dagger) \Big\rVert_1 \nonumber \\
	\label{eq:smoothing second}
		&+ \EE \bigl\lVert \mathcal D_{BB_1\to A_0} \circ \overline{\mathcal T}_{A_0A_1\to B}(\Phi^+_{A_0R} \otimes \Phi^+_{A_1B_1}) \\
    &\qquad - \Phi^+_{A_0R} \bigr\rVert_1 \nonumber,
  \end{align}
  where $\widehat{\mathcal T}^i_{A\to B}(\sigma_A) := \mathcal N^i_{A\to B}\left(d_A O_A\left(\rho^i\right)\sigma_AO_A\left(\rho^i\right)^\dagger\right)$. In the first inequality, we have inserted~\eqref{eq:IS ansatz}--\eqref{eq:average general} and used the triangle inequality as well as that $\mathcal D_{BB_1\to A_0}$ is trace-nonincreasing.

  To bound the averages in~\eqref{eq:smoothing first}, we now follow the smoothing ideas from~\cite{dupuis2014one}. Let $\widehat\rho_{A'B}^i := \widehat{\mathcal T}_{A\to B}^i(\Phi^+_{AA'}) = \mathcal N^i_{A\to B}(\rho_{AA'}^i)$ and write $\widehat\rho_{A'B}^i - \widetilde\rho_{A'B}^i = \delta^{i,+}_{A'B} - \delta^{i,-}_{A'B}$ as a difference of positive semidefinite operators, so that
  \begin{equation}\label{eq:delta pm bound}
  \begin{aligned}
    &\quad\tr\left[\delta^{i,+}_{A'B}\right] + \tr\left[\delta^{i,-}_{A'B}\right] = \lVert \widehat\rho_{A'B}^i - \widetilde\rho_{A'B}^i \rVert_1 \\
    &\leq 2 P(\widehat\rho_{A'B}^i, \widetilde\rho_{A'B}^i) \leq 2 \eps
  \end{aligned}
  \end{equation}
  by~\eqref{eq:purified distance and trace norm subnormalized} together with~\eqref{eq:desmoothing}.
  Let $\mathcal D^{i,\pm}_{A\to B}$ denote completely positive maps whose Choi-Jamiolkowski states are $\delta^{i,\pm}_{A'B}$, respectively. Then, we have $\widehat{\mathcal T}_{A\to B}^i - \mathcal T^i_{A\to B} = \mathcal D^{i,+}_{A\to B} - \mathcal D^{i,-}_{A\to B}$, and hence
  \begin{align*}
	&\quad\EE \left\lVert \bigl( \widehat{\mathcal T}^i_{A\to B} - \mathcal T^i_{A\to B} \bigr)(U^i J^i (\Phi^+_{A_0R} \otimes \Phi^+_{A_1^iB_1^i}) (J^i)^\dagger (U^i)^\dagger) \right\rVert_1 \\
	&\leq \EE \tr \left[\mathcal D^{i,+}_{A\to B}(U_A^i J^i (\Phi^+_{A_0R} \otimes \Phi^+_{A_1^iB_1^i}) (J^i)^\dagger (U_A^i)^\dagger)\right] \\
	&\quad+ \EE \tr\left[\mathcal D^{i,-}_{A\to B}(U_A^i J^i (\Phi^+_{A_0R} \otimes \Phi^+_{A_1^iB_1^i}) (J^i)^\dagger (U_A^i)^\dagger)\right] \\
	&\leq \tr\left[\delta_B^{i,+} \otimes \tau_{RB_1^i}\right] + \tr\left[\delta_B^{i,-} \otimes \tau_{RB_1^i}\right]\\
	&\leq 2 \eps,
  \end{align*}
  by the triangle inequality, the fact that $\EE(U_A^i \sigma_{AB} (U_A^i)^\dagger) = \tau_A \otimes \tr_A[\sigma_{AB}]$ for all $\sigma_{AB}$, and~\eqref{eq:delta pm bound}.

  To bound the average in~\eqref{eq:smoothing second}, we use~\eqref{eq:uhlmann decoder bound}, Jensen's inequality and \cref{lem:average decoupling non-smooth} to obtain
  \begin{align*}
		&\quad \EE \left\lVert \mathcal D_{BB_1\to A_0} \circ \overline{\mathcal T}_{A_0A_1\to B}(\Phi^+_{A_0R} \otimes \Phi^+_{A_1B_1}) - \Phi^+_{A_0R} \right\rVert_1 \\
		&\leq (\EE\Delta) + 2\sqrt{2\EE\Delta} \leq \delta + 2\sqrt{2\delta}
  \end{align*}
  with
  \begin{align*}
    \delta
    = 2^{-\frac12 \bigl( -\max_i \Hmax^\eps(A'|B)_{\mathcal N^i(\rho^i)} - \log M_0 + \log M_1^i - 2 \log N - 2 \bigr)},
  \end{align*}
  where we used that $H_2(A'|E)_{\mathcal T^{i,c}(\Phi^+)} \geq \Hmin(A'|E)_{\mathcal T^{i,c}(\Phi^+)}$ and~\eqref{eq:desmoothing complete}. Plugging both bounds into~\eqref{eq:smoothing todo}, we obtain the desired estimate.
\end{proof}

The decoding maps $\mathcal E^i_{A_0A_1\to A}$ as defined in~\eqref{eq:IS ansatz} are completely positive but not in general trace-preserving, and therefore not valid quantum operations. However, the following lemma, whose proof is entailed in~\cite[Theorem 3.14]{dupuis2009decoupling} and which can be deduced directly from the one-shot decoupling theorem~\cite[Theorem 3.3]{dupuis2014one}, will later allow us to replace the $\mathcal E^i$ by valid quantum operations.

\begin{lemma}\label{lem:complement unchanged}
  Let $\mathcal E^i_{A_0A_1\to A}$ be one of the completely positive maps defined in~\eqref{eq:IS ansatz} for some state $\rho_A^i$.
  Then, we have
  \begin{align*}
    &\quad\EE \left\lVert \tr_A\left[\mathcal E^i_{A_0A_1\to A}(\Phi^+_{A_0R} \otimes \Phi^+_{A_1B_1})\right] - \tau_{RB_1} \right\rVert_1 \nonumber \\
    &\leq 2^{-\frac12(\Hmin^\eps(A)_{\rho^i} - \log M_0 - \log M_1^i)} + 12 \eps
  \end{align*}
\end{lemma}

By combining this with \cref{lem:average decoupling smooth}, we obtain our one-shot coding theorem for compound channels in the informed sender scenario.

\begin{theorem}\label{thm:IS one-shot}
  Let $\Pi_{A\to B} = \{\mathcal N^i_{A\to B}\}_{i=1}^N$ be a finite compound channel, $\{\rho^i_{AA'}\}_{i=1}^N$ pure states, and $M_0$, $\{M_1^i\}_{i=1}^N$ integers such that $M_0 M_1^i \leq d_A$ for all $i$, and $\eps \in (0,1]$.
  Then, there exist quantum operations $\mathcal E^i_{A_0A_1\to A}$ and $\mathcal D_{BB_1\to A_0}$, where $d_{A_0} = M_0$ and $d_{A_1} = d_{B_1} \geq \max_i M_1^i$, such that
  \begin{equation}\label{eq:claim of theorem}
  \begin{aligned}
  &\quad\min_i F(\Phi^+_{A_0R},\mathcal D \circ \mathcal N^i \circ \mathcal E^i(\Phi^+_{A_0R}\otimes\Phi^+_{A_1 B_1})) \\
  &\geq 1 - 8 N (N+2) \bigl(\sqrt{\delta_1} + \sqrt{\delta_2} + 6 \sqrt{\eps} \bigr),
  \end{aligned}
  \end{equation}
  where
  \begin{align*}
    \delta_1 &= \max_i 2^{-\frac12(\Hmin^\eps(A)_{\rho^i} - \log M_0 - \log M_1^i)}, \\
  	\delta_2 &= \max_i 2^{-\frac12 \bigl( -\Hmax^\eps(A'|B)_{\mathcal N^i(\rho^i)} - \log M_0 + \log M_1^i - 2 \log N - 2 \bigr)}.
  \end{align*}
\end{theorem}

\begin{proof}
  We choose $M_1$ as the least common multiple of the $M_1^i$. By \cref{lem:average decoupling smooth,lem:complement unchanged}, Markov's inequality $\Prob(Z > k \EE(Z)) \leq 1/k$ (for $k=N+2$), and the union bound, there exist unitaries $\{U_A^i\}_{i=1}^N$ and a quantum operation $\mathcal D_{BB_1\to A_0}$ such that
  \begin{align*}
    &\quad \left\lVert \tr_A\left[\mathcal E^i_{A_0A_1\to A}(\Phi^+_{A_0R} \otimes \Phi^+_{A_1B_1})\right] - \tau_{RB_1} \right\rVert_1 \\
  	&\leq (N+2) (\delta_1 + 12\eps) \quad (\forall i=1,\dots,N)
  \end{align*}
  with
  \begin{align*}
		&\Big\lVert (\mathcal D_{BB_1\to A_0} \circ \frac1N \sum_{i=1}^N \mathcal N^i_{A\to B} \circ \mathcal E^i_{A_0A_1\to A})(\Phi^+_{A_0R} \otimes \Phi^+_{A_1B_1})  \nonumber \\
		&-\Phi^+_{A_0R} \Big\rVert_1
		\leq (N+2)(\delta_2 + 2\sqrt{2\delta_2} + 2 \eps).
  \end{align*}
  As a consequence of the first bound and \cref{lem:uhlmann one is normalized}, we can find quantum operations $\{ \widetilde{\mathcal E}^i_{A_0A_1\to A} \}_{i=1}^N$ such that
  \begin{align*}
  &\quad\left\lVert (\widetilde{\mathcal E}^i_{A_0A_1\to A} - \mathcal E^i_{A_0A_1\to A})(\Phi^+_{A_0R} \otimes \Phi^+_{A_1B_1}) \right\rVert_1 \nonumber \\
   & \leq (N+2) (\delta_1 + 12\eps) + 2 \sqrt{2 (N+2) (\delta_1 + 12\eps)}.
  \end{align*}
  Now, using the triangle inequality as well as the fact that $\mathcal D$ and the $\mathcal N^i$ are completely positive and trace-nonincreasing, we obtain that
   \begin{align*}
  	&\quad\Big\lVert \frac1N \sum_{i=1}^N (\mathcal D \circ \mathcal N^i \circ \widetilde{\mathcal E}^i)(\Phi^+_{A_0R} \otimes \Phi^+_{A_1B_1}) - \Phi^+_{A_0R}\Big\rVert_1 \\
  	&\leq \frac1N \sum_{i=1}^N \left\lVert (\widetilde{\mathcal E}^i_{A_0A_1\to A} - \mathcal E^i_{A_0A_1\to A})(\Phi^+_{A_0R} \otimes \Phi^+_{A_1B_1}) \right\rVert_1 \\
  	&\quad+ \Big\lVert (\mathcal D \circ \frac1N \sum_{i=1}^N \mathcal N^i \circ \mathcal E^i)(\Phi^+_{A_0R} \otimes \Phi^+_{A_1B_1})- \Phi^+_{A_0R} \Big\rVert_1 \\
  	&\leq(N+2) (\delta_1 + 12\eps) + 2 \sqrt{2 (N+2) (\delta_1 + 12\eps)} \\
	&\quad + (N+2) (\delta_2 + 2\sqrt{2\delta_2} + 2 \eps) \\
	&\leq 4 (N+2) (\sqrt{\delta_1} + \sqrt{\delta_2} + 6\sqrt{\eps}).
	\end{align*}
  To arrive at the last inequality, we have assumed that $\delta_{1,2} \leq 1$ (without loss of generality, since the bound~\eqref{eq:claim of theorem} is otherwise vacuous). At last, we use~\eqref{eq:fidelity lower bound one is normalized} to turn this into a lower bound on the average entanglement fidelity:
  \begin{align*}
  &\quad F(\Phi^+_{A_0R}, (\frac1N \sum_{i=1}^N \mathcal D \circ \mathcal N^i \circ \widetilde{\mathcal E}^i)(\Phi^+_{A_0R} \otimes \Phi^+_{A_1B_1})) \\
  &\geq 1 - 8 (N+2) \bigl(\sqrt{\delta_1} + \sqrt{\delta_2} + 6 \sqrt{\eps} \bigr).
  \end{align*}
  Using the same argument that we used to derive~\eqref{eq:entanglement fidelity union bound}, this implies that
  \begin{align*}
  &\quad\min_i F(\Phi^+_{A_0R}, (\mathcal D \circ \mathcal N^i \circ \widetilde{\mathcal E}^i)(\Phi^+_{A_0R} \otimes \Phi^+_{A_1B_1})) \\
  &\geq 1 - 8 N (N+2) \bigl(\sqrt{\delta_1} + \sqrt{\delta_2} + 6 \sqrt{\eps} \bigr). \qedhere
  \end{align*}
\end{proof}

\subsection{Achievability}

Given the one-shot coding theorem phrased in terms of smooth entropies (\cref{thm:IS one-shot}), we can now prove the direct part of \cref{thm:IS} in a similar fashion as for \cref{thm:uninformed}:

\begin{lemma}\label{lem:IS direct}
	Let $\Pi_{A \rightarrow B}  = \{ \mathcal N^i_{A\rightarrow B} \}_{i \in I}$ be a compound channel with arbitrary index set $I$.
	Then, we have
	\begin{align*}
		Q_{E,IS}(\Pi_{A \rightarrow B}) \geq \inf_{i \in I} Q_E(\mathcal N^i_{A\to B}).
	\end{align*}
\end{lemma}

\begin{proof}
  Let $\Delta > 0$ and
  \begin{align*}
  R = \inf_{i \in I} Q_E(\mathcal N^i_{A\to B}) - \Delta.
  \end{align*}
  We will show that for any $\delta > 0$ there exists $n_0$ such that, for all $n \geq n_0$, the triple $(R,n,\delta)$ is achievable for $\Pi_{A\to B}$.
  If $R \leq 0$ then there is nothing to show, thus we may assume that $R > 0$.

  As in the proof of \cref{lem:uninformed direct}, we will use the discretization result~\cite[Lemma V.2]{bjelakovic2008quantum} to reduce to finite compound channels.
  It shows that for any $n$ there exists a finite compound channel $\widetilde\Pi_{A\to B} = \{ \widetilde{\mathcal N}^j_{A\to B} \}_{j=1}^N$ of cardinality
  $N \leq (3 n)^{4 d_{AB}^2}$
  with the property that for any $\mathcal N^i \in \Pi$ there exists some $\widetilde{\mathcal N}^{j} \in \widetilde\Pi$ such that
  $\lVert (\mathcal N^i)^{\otimes k} - (\widetilde{\mathcal N}^{j})^{\otimes k} \rVert_\diamond \leq k/n^2$ for all $k$, and vice versa. In particular, for $k=1$ this bound together with~\eqref{eq:fannes mutual info} implies that
  \begin{align*}
  &\quad\bigl| \inf_{i \in I} Q_E(\mathcal N^i) - \min_{j=1,\dots,N} Q_E(\widetilde{\mathcal N}^j) \bigr| \nonumber \\
  &\leq \frac 1 {n^2} \log\min\{d_A,d_B\} + 2\left(1+\frac{1}{2n^2}\right)h\biggl( \frac 1 {1+2n^2}\biggr).
  \end{align*}
  As a consequence,
  \begin{align}\label{eq:R vs discretization IS}
    R \leq \min_{j=1,\dots,N} Q_E(\widetilde{\mathcal N}^j) - \frac\Delta2
  \end{align}
  for $n$ sufficiently large (depending only on $d_{AB}$ and $\Delta$). Let us assume that this is the case.

  Let $\{ \rho_{AA'}^j \}$ be pure states such that $Q_E(\widetilde{\mathcal N}^j) = \frac12 I(A':B)_{\widetilde{\mathcal N}^j(\rho^j)}$ for all $j=1,\dots,N$.
  We now apply our one-shot coding result, \cref{thm:IS one-shot}, to $\widetilde\Pi^{\otimes n}_{A\to B} = \{ (\widetilde{\mathcal N}^j_{A\to B})^{\otimes n} \}_{j=1}^N$, $\{ (\rho_{AA'}^j)^{\otimes n} \}_{j=1}^N$, $M_0 = \lceil 2^{nR} \rceil$, $M_1^j = \lceil 2^{n(H(A)_{\rho^j} - R - \Delta/2)} \rceil$.
  We note that
  \begin{align*}
  	M_0 M_1^j&\leq 2^{nR + 1} 2^{n(H(A)_{\rho^j} - R - \Delta/2) + 1} \nonumber \\
	&= 2^{n(H(A)_{\rho^j} - \Delta/2 + 2/n)} \leq d_{A^n}
  \end{align*}
  for $n$ sufficiently large (depending only on $R$ and $\Delta$). Thus, the assumption on the integers $M_0$, $\{M_1^j\}$ is satisfied. For all $\eps \in (0,1]$, we obtain encoders $\widetilde{\mathcal E}_{A_0A_1\to A^n}^j$ and a decoder $\mathcal D_{B^nB_1\to A_0}$ such that
  \begin{equation}\label{eq:IS F discretized}
  \begin{aligned}
    &\quad\min_j F(\Phi^+_{A_0R}, \mathcal D \circ (\widetilde{\mathcal N}^j)^{\otimes n} \circ \widetilde{\mathcal E}^j(\Phi^+_{A_0R}\otimes\Phi^+_{A_1 B_1})) \\
     &\geq 1 - 8 N (N+2) \bigl(\sqrt{\delta_1} + \sqrt{\delta_2} + 6 \sqrt{\eps} \bigr),
  \end{aligned}
  \end{equation}
  where
  \begin{align*}
    \delta_1 &= \max_j 2^{-\frac12(\Hmin^\eps(A)_{(\rho^j)^{\otimes n}} - \log M_0 - \log M_1^j)}, \\
  	\delta_2 &= \max_j 2^{-\frac12 \bigl( -\Hmax^\eps(A'|B)_{(\widetilde{\mathcal N}^j(\rho^j))^{\otimes n}} - \log M_0 + \log M_1^j - 2 \log N - 2 \bigr)}\!.
  \end{align*}
  We now choose $\eps = 1/(n N)^4$.
  Since $\eps$ decays only polynomially with $n$, the asymptotic equipartition property \eqref{eq:AEP H_min} together with~\eqref{eq:AEP convergence} and the estimate~\eqref{eq:R vs discretization IS} implies that
  \begin{align*}
    \delta_1
    \leq 2^{-\frac n 2 \bigl( \frac \Delta 2 - \frac 2 n - \frac {\delta(\eps,\rho^j)} {\sqrt n} \bigr)}
    \leq 2^{-n \frac \Delta 8}
    \leq \eps
  \end{align*}
  for sufficiently large $n$ (depending only $d_{AB}$, $\Delta$ and $R$).
  Likewise, using \eqref{eq:AEP H_max} instead of \eqref{eq:AEP H_min} we obtain that
  \begin{align*}
    \delta_2
    \leq 2 N \cdot 2^{-\frac n 2 \bigl(\frac \Delta 2 - \frac 1 n - \frac {\delta(\eps,\widetilde{\mathcal N}^j(\rho^j))} {\sqrt n} \bigr)}
    \leq 2 N \cdot 2^{-n \frac \Delta 8}
    \leq \eps,
  \end{align*}
  where we use that $N$ grows only polynomially with $n$. By inserting the two bounds into~\eqref{eq:IS F discretized}, we obtain that
  \begin{align*}
    &\quad F(\Phi^+_{A_0R}, \mathcal D \circ (\widetilde{\mathcal N}^j)^{\otimes n} \circ \widetilde{\mathcal E}^j(\Phi^+_{A_0R}\otimes\Phi^+_{A_1 B_1})) \nonumber \\
    &\geq 1 - 64 N (N+2) \sqrt{\eps}
    \geq 1 - \frac {192} {n^2}.
  \end{align*}
  At last, we relate this to the entanglement fidelity for the original compound channel. For this, recall that for any $\mathcal N^i \in \Pi$ there exists some $\widetilde{\mathcal N}^{j} \in \widetilde\Pi$ such that $\lVert (\mathcal N^i)^{\otimes n} - (\widetilde{\mathcal N}^{j})^{\otimes n} \rVert_\diamond \leq 1/n$.
  If we choose the encoders correspondingly as $\mathcal E^i := \widetilde{\mathcal E}^j$ then we find using~\eqref{eq:fidelity channel continuity} that
  \begin{align*}
    &\quad\inf_{i \in I} F(\Phi^+_{A_0R},\mathcal D \circ (\mathcal N^i)^{\otimes n} \circ \mathcal E^i(\Phi^+_{A_0R}\otimes\Phi^+_{A_1 B_1})) \nonumber \\
    &\geq 1 - \frac {192} {n^2} - \frac 2 {\sqrt n}
  \end{align*}
  We conclude that, for any $\delta > 0$ and sufficiently large $n$, $(R,\delta,n)$ is a valid triple for the compound channel $\Pi_{A\to B}$.
\end{proof}

\subsection{Converse}

Since a code for the compound quantum channel $\Pi_{A\to B} = \{ \mathcal N^i_{A\to B} \}$ by definition gives rise to codes for each of its constituent channels $\mathcal N^i_{A\to B}$, it is immediate that $Q_{E,IS}(\Pi) \leq Q_E(\mathcal N^i)$ for all $i \in I$. Thus, we immediately obtain the converse bound in \cref{thm:IS}.

\begin{lemma}\label{lem:IS converse}
Let $\Pi_{A \rightarrow B}  = \{ \mathcal N^i_{A\rightarrow B} \}_{i \in I}$ be a compound channel with arbitrary index set $I$.
Then, we have
\begin{align*}
Q_{E,IS}(\Pi_{A \rightarrow B}) \leq \inf_{i\in I}Q_E(\mathcal{N}^i_{A\to B}).
\end{align*}
\end{lemma}

\section{Feedback assistance}

It is well-known that feedback does not increase the entanglement-assisted quantum capacity of a quantum channel~\cite{bowen02}. Since any feedback-assisted code for the compound channel gives rise to feedback-assisted codes for each of its constituent channels, we obtain just as in the preceding section the converse bound in \cref{thm:feedback}. In fact, this holds for arbitrary compound channels:

\begin{lemma}\label{feedback:converse}
Let $\Pi_{A \rightarrow B}  = \{ \mathcal N^i_{A\rightarrow B} \}_{i \in I}$ be a compound channel with arbitrary index set. Then, we have
\begin{align*}
Q_{E,F}(\Pi_{A \rightarrow B}) \leq \inf_{i\in I}Q_E(\mathcal{N}^i_{A\to B}).
\end{align*}
\end{lemma}

Next we show that the upper bound in \cref{feedback:converse} is also achievable, at least for finite compound channels (establishing \cref{thm:feedback}):

\begin{lemma}\label{feedback:feedback direct}
Let $\Pi_{A \rightarrow B}  = \{ \mathcal N^i_{A\rightarrow B} \}_{i \in I}$ be a compound channel with finite index set $\lvert I \rvert < \infty$. Then, we have
\begin{align*}
Q_{E,F}(\Pi_{A\to B}) \geq \inf_{i\in I}Q_E(\mathcal{N}^i_{A\to B}).
\end{align*}
\end{lemma}

The proof is a generalization of the original proof for classical channels~\cite{wolfowitz78} and based on the following quantum channel estimation technique from~\cite{1751-8121-40-28-S20}:

\begin{proposition}[{\cite[Theorem 4.2]{bjelakovic2008quantum}}]\label{lem:channel_estimation}
Let $\Pi_{A\to B}=\left\{\mathcal{N}_{A\to B}^i\right\}_{i=1}^{N}$ be a finite compound channel and set $L=\binom{N}{2}$. Then, there exists $f\in(0,1)$ such that for each $m\in\mathbb{N}$ there are mutually orthogonal projectors $\left\{P^i_{B^{mL}}\right\}_{i=1}^N$ on $B^{\otimes(mL)}$ with $\sum_{i=1}^NP^i_{B^{mL}}=I_{B^{mL}}$, as well as a pure state $\omega_{A^{mL}}$ on $A^{\otimes(mL)}$ with the property that
\begin{align*}
\tr\left[P^i_{B^{mL}}\left(\mathcal{N}^i_{A\to B}\right)^{\otimes(mL)}\left(\omega_{A^{mL}}\right)\right]\geq\left(1-Nf^m\right)^{N-1}
\end{align*}
for all $i=1,\dots,N$.
\end{proposition}

For the IID compound $\Pi^{\otimes n}_{A\to B}=\left\{\left(\mathcal{N}_{A\to B}^i\right)^{\otimes n}\right\}_{i\in I}$ we use the first $\sqrt{n}$ channel uses to estimate the channel on the receiver's side with the help of \cref{lem:channel_estimation}. Then, we use feedback to transfer the estimated channel index $i\in I$ to the sender. This allows to use the informed sender protocol for the remaining $n-\sqrt{n}$ channel uses and leads to the same capacity as in the informed sender case.
We formalize this strategy in the following proof, which is inspired by~\cite[Lemma 4.3]{bjelakovic2008quantum}:

\begin{proof}[Proof of \cref{feedback:feedback direct}]
We give a protocol for the IID average channel
\begin{align*}
\overline{\Pi}_{A^{n}\to B^{n}}:=\frac{1}{|I|}\sum_{i\in I}\left(\mathcal N^i_{A\to B}\right)^{\otimes n},
\end{align*}
which will work equally well for the IID compound channel (up to a constant factor of $|I|$ in fidelity, cf.~\eqref{eq:entanglement fidelity union bound}). Take $n=mL+t$ with $L:=\binom{|I|}{2}$ and $m\in\mathbb{N}$. We use the first $mL$ instances for channel estimation and the subsequent $t=n-mL$ instances for the actual entanglement transmission, making use of the informed sender protocol as described in \cref{thm:IS one-shot}.

The encoder $\tilde{\mathcal{E}}_{A^{mL}}$ for the channel estimation inputs $\omega_{A^{mL}}$ from \cref{lem:channel_estimation} to the $(mL)$-fold channel. The decoder $\tilde{\mathcal{D}}_{B^{mL}}$ for the channel estimation measures the channel's output as in \cref{lem:channel_estimation},
\begin{align*}
\tilde{\mathcal{D}}_{B^{mL}}(\cdot):=\sum_{i\in I}\tilde{\mathcal{D}}^i_{B^{mL}}(\cdot)\ket{i}\!\!\bra{i}_{X_{B}}
\end{align*}
with $\tilde{\mathcal{D}}^i_{B^{mL}}(\cdot):=\tr\left[P^i_{B^{mL}}(\cdot)\right]$,
and sends the classical outcome $i\in I$ back to the sender, where it is then labeled by the system $X_A$.
Given this information $i\in I$ we use for the entanglement transmission the informed encoder $\mathcal{E}^i_{A_0A_1\to A^{t}}$ for $\left(\mathcal N^i_{A\to B}\right)^{\otimes t}$, i.e.,
\begin{align*}
\mathcal{E}_{A_0A_1X_A\to A^{t}}(\cdot):=\sum_{i\in I}\mathcal{E}^i_{A_0A_1\to A^{t}}(\cdot)\otimes \braket{i | \cdot | i}_{X_A},
\end{align*}
and the universal decoder $\mathcal{D}_{B^{t}B_1\to A_0}$ from \cref{thm:IS one-shot}.\footnote{Since both the sender and the receiver know the classical outcome $i\in I$ we could alternatively simply use an encoder and decoder obtained by coding theorems for fixed channels (see, e.g., \cite[Theorem 3.14, cf.~Theorem 3.15]{dupuis2009decoupling}).} The total fidelity of the protocol can then be bounded as
\begin{align}
&\quad\sum_{i\in I}F\Big(\Phi^+_{A_0R},\bigl( (\tilde{\mathcal{D}}^i_{B^{mL}}\otimes\mathcal{D}_{B^{t}B_1\to A_0}) \circ \bar{\Pi}_{A^{mL+t}\to B^{mL+t}}  \nonumber \\
&\qquad\qquad\qquad \circ (\tilde{\mathcal{E}}_{A^{mL}}\otimes\mathcal{E}^i_{A_0A_1\to A^{t}})\bigr) (\Phi^+_{A_0R}\otimes\Phi^+_{A_1B_1})\Big)\notag\\
&\geq \frac{1}{|I|}\sum_{i\in I}F\Big(\Phi^+_{A_0R},\bigl( (\tilde{\mathcal{D}}^i_{B^{mL}}\otimes\mathcal{D}_{B^{t}B_1\to A_0}) \circ (\mathcal{N}^{i}_{A\to B})^{\otimes(mL+t)}\nonumber \\
&\qquad\qquad\qquad \circ(\tilde{\mathcal{E}}_{A^{mL}}\otimes\mathcal{E}^i_{A_0A_1\to A^{t}}) \bigr) (\Phi^+_{A_0R}\otimes\Phi^+_{A_1B_1})\Big)\notag\\
&=\frac{1}{|I|}\sum_{i\in I}F\Big(\Phi^+_{A_0R}, \bigl(\mathcal{D}_{B^{t}B_1\to A_0}\circ(\mathcal{N}^{i}_{A\to B})^{\otimes t} \notag\\
&\qquad\qquad\qquad \circ\mathcal{E}^i_{A_0A_1\to A^{t}}\bigr) (\Phi^+_{A_0R}\otimes\Phi^+_{A_1B_1})\Big)\notag\\
&\qquad\quad \cdot \tr\left[P^i_{B^{mL}}\left(\mathcal{N}^i_{A\to B}\right)^{\otimes(mL)}\left(\omega_{A^{mL}}\right)\right]\notag\\
&\geq(1-|I|f^m)^{|I|-1}(1-\delta),\label{eq:error_feedback}
\end{align}
for $f\in(0,1)$ as in \cref{lem:channel_estimation}. In the last inequality we have assumed that the informed encoders $\left\{\mathcal{E}^i_{A_0A_1\to A^{t}}\right\}_{i\in I}$ together with the universal decoder $\mathcal{D}_{B^{t}B_1\to A_0}$ have a fidelity of at least $1-\delta$ for $\delta>0$ (cf.~\cref{thm:IS one-shot}). Now, we choose $m=\lfloor\sqrt{n}\rfloor$ and $t=n-\lfloor\sqrt{n}\rfloor L$. For $n\to\infty$ the total error in~\eqref{eq:error_feedback} then tends to $1-\delta$. Moreover, we find by \cref{lem:IS direct} that for $n\to\infty$ we can transmit entanglement at any rate $R\leq\inf_{i\in I}Q_E(\mathcal{N}^i_{A\to B})$:
\begin{align*}
\frac{1}{n}\log M_0\geq\frac{t\cdot R}{n}=\frac{n\cdot R-\lfloor\sqrt{n}\rfloor LR}{n}
\to R.
\end{align*}
Since $\delta>0$ was arbitrary, this concludes the proof.
\end{proof}

\section{Conclusion}\label{sec:conclusion}

In this article, we have determined the entanglement-assisted quantum capacity of compound quantum channels for various setups. In particular, we have provided closed formulas for the entanglement-assisted capacity for uninformed users, for an informed receiver, for an informed sender, and in the presence of free feedback from the receiver to the sender. All our findings are in complete analogy to the case of classical compound channels and hence strengthen the interpretation of entanglement-assisted communication as the most tractable generalization of classical information theory to the quantum setting.

Our proofs are based on one-shot decoupling theorems, properties of smooth entropies, and make use of some previously developed tools for analyzing compound quantum channels~\cite{boche09,hayashi09b,Boche13,1751-8121-40-28-S20,Datta2008,bjelakovic2008quantum,bjelakovic09,bjelakovic2009entanglement}. We believe that our approach may also lead to an improved understanding of capacities of compound quantum channels in general. As an illustration we present in \cref{app:plain IS} a simplified argument for studying the plain quantum capacity of an arbitrary compound quantum channel with an informed sender (cf.~the original works~\cite{bjelakovic2008quantum,bjelakovic09,bjelakovic2009entanglement}).

We end with a discussion of a few open questions. In the feedback-assisted scenario, we were only able to determine the capacity of finite compound quantum channels. It is not known how to extend this to arbitrary compounds and the corresponding solution for classical compounds might serve as a good starting point~\cite{Shrader09}. Moreover, it would be interesting for all the setups discussed in our work to optimize the amount of entanglement assistance that is needed, and with that to characterize the whole rate region $(M_0,M_1)$. Finally, a variant of compound quantum channels known as arbitrarily varying quantum channels (AVC) have been studied in the literature (see, e.g., \cite{ahlswede10,Ahlswede2012,Boche13,noetzel13}). Here, for a fixed set of channels $\left\{\mathcal{N}^i\right\}_{i\in I}$ the goal is to find protocols for information transmission that work reliable for all channels of the form
\begin{align*}
\mathcal{N}^{(i_1,\ldots,i_n)}:=\otimes_{j=1}^n\mathcal{N}^{i_j}\quad\text{in the limit $n\to\infty$.}
\end{align*}
The resulting entanglement-assisted arbitrarily varying quantum capacity is not known, but by analogy with the classical results~\cite{blackwell60,stiglitz66,ahlswede69} one naturally conjectures the following closed formula:
\begin{align}\label{eq:avqc}
Q_{E}\left(\left\{\mathcal{N}^i\right\}_{i\in I}\right)=\inf_{\mathcal{N}\in\text{conv}\left(\left\{\mathcal{N}^i\right\}_{i\in I}\right)}Q_{E}\left(\mathcal{N}\right)
\end{align}
We expect the methods developed in this article to be useful for attacking the conjecture.

\smallskip

\emph{Note added:} After completion of our work the conjecture in~\eqref{eq:avqc} was proven in~\cite{boche16paper}.


\section*{Acknowledgments}

We thank David Ding, Patrick Hayden, Joseph Renes, and Volkher Scholz for helpful discussions, and Mark Wilde for pointing out  the tight continuity estimate~\eqref{eq:example_continuity} (see also~\cite{wildebook13}).


\appendix

\section{Proofs of technical lemmas}\label{app:smooth entropies}

\begin{proof}[Proof of \cref{lem:uhlmann one is normalized}]
  Let $\rho'_{AB} := \rho_{AB} / \tr[\rho_{AB}]$ denote the normalized version of $\rho_{AB}$. Then, we have
  \begin{align*}
  &\quad \lVert \rho'_A - \rho_A \rVert_1 = \lVert \rho'_{AB} - \rho_{AB} \rVert_1 = \lvert 1 - \tr[\rho_A] \rvert  \nonumber \\
  &= \lvert \tr [\rho_A - \sigma_A] \rvert \leq \lVert \rho_A - \sigma_A \rVert_1 \leq \delta,
  \end{align*}
  and thus $\lVert \rho'_A - \sigma_A \rVert_1 \leq 2 \delta$.
  From the first inequality in~\eqref{eq:fuchs-van de graaf}, we obtain that $F(\rho'_A, \sigma_A) \geq (1 - \delta)^2$.
  Thus, by Uhlmann's theorem~\eqref{eq:uhlmann fred} there exist partial isometries $V_{B\to C}$ and $W_{C\to B}$ such that
  \begin{align*}
  &\quad F(V_{B\to C} \rho'_{AB} V_{B\to C}^\dagger, \sigma_{AC}) = F(\rho'_{AB}, W_{C\to B} \sigma_{AC} W_{C\to B}^\dagger)\\
  &\geq (1 - \delta)^2 \geq 1 - 2\delta.
  \end{align*}
  From the second inequality in~\eqref{eq:fuchs-van de graaf}, we thus obtain that
  $\lVert V_{B\to C} \rho'_{AB} V_{B\to C}^\dagger - \sigma_{AC} \rVert_1 \leq 2 \sqrt{2 \delta}$
  and hence
  \begin{align*}
    &\quad \lVert V_{B\to C} \rho_{AB} V_{B\to C}^\dagger - \sigma_{AC} \rVert_1 \nonumber \\
    &\leq \lVert \rho'_{AB} - \rho_{AB} \rVert_1 + 2\sqrt{2\delta}
    \leq \delta + 2\sqrt{2\delta}
  \end{align*}
  Similarly, $\lVert \rho'_{AB} - W_{C\to B} \sigma_{AC} W_{C\to B}^\dagger \rVert_1 \leq 2 \sqrt{2 \delta}$ and thus
  \begin{equation*}
  \lVert \rho_{AB} - W_{C\to B} \sigma_{AC} W_{C\to B}^\dagger \rVert_1 \leq \delta + 2 \sqrt2\delta. \qedhere
  \end{equation*}
\end{proof}

\medskip

\begin{proof}[Proof of \cref{lem:min entropy mixture}]
  According to \eqref{eq:min entropy sdp}--\eqref{eq:smooth min entropy}, we can find sub-normalized states $\widetilde\rho^i_{AB} \in \mathcal S_\leq(AB)$ as well as $\sigma_B^i \in \mathcal P(B)$ such that $P(\rho^i, \widetilde\rho^i) \leq \eps$, $\widetilde\rho_{AB}^i \leq I_A \otimes \sigma_B^i$, and $2^{-\Hmin^\eps(A|B)_{\rho^i}} = \tr\left[\sigma_B^i\right]$ for all $i=1,\dots,N$. Now, we define the sub-normalized state $\widetilde\rho_{AB} := \sum_{i=1}^N p_i \widetilde\rho_{AB}^i$. The joint quasi-convexity of the purified distance~\cite[(3.60)]{tomamichel2015quantum} implies that $P(\widetilde\rho_{AB}, \rho_{AB}) \leq \eps$, and so~\eqref{eq:smooth min entropy} shows that
  \begin{align*}
    \Hmin^\eps(A|B)_\rho \geq \Hmin(A|B)_{\widetilde\rho}.
  \end{align*}
  In order to lower-bound the right-hand side, we consider $\sigma_B := \sum_{i=1}^N p_i \sigma_B^i$.
  Clearly, $\widetilde\rho_{AB} = \sum_{i=1}^N p_i \widetilde\rho_{AB}^i \leq \sum_{i=1}^N p_i I_A \otimes \sigma_{B}^i = I_A \otimes \sigma_B$. Thus, $\sigma_B$ is a feasible point for the optimization in~\eqref{eq:min entropy sdp}, and so
	\begin{align*}
		&\quad \Hmin(A|B)_{\widetilde\rho} \geq -\log \tr[\sigma_B] = -\log \left( \sum_{i=1}^N p_i \tr\left[\sigma_B^i\right] \right)\\
		&\geq \min_{i} -\log \tr\left[\sigma_B^i\right]= \min_{i} \Hmin^\eps(A|B)_{\rho^i},
	\end{align*}
	using the quasi-concavity of $x \mapsto -\log x$.
\end{proof}

\medskip

\begin{proof}[Proof of \cref{lem:max entropy mixture}]
  For all $i=1,\dots,N$, let $\ket{\rho_{ABC}^i}$ a purification of $\rho_{AB}^i$.
  Then, we have that
  \begin{align*}
    \ket{\rho_{ABCI_1I_2}} := \frac1{\sqrt N} \sum_{i=1}^N \ket{\rho_{ABC}^i} \otimes \ket{i_{I_1}}\ket{i_{I_2}}
  \end{align*}
  is a purification of $\rho_{AB}$. We note that both $\rho_{ABI_1}$ and $\rho_{ACI_2}$ are classical on $I_1$ and $I_2$, respectively. Thus, we can apply~\cite[Lemma 6.8]{tomamichel2015quantum}, which, together with~\eqref{eq:smooth max entropy duality} to switch from max- to min-entropy, shows that
  \begin{align*}
    \Hmax^\eps(A|B)_\rho &\leq \Hmax^\eps(A|BI_1)_\rho + \log N \nonumber \\
    &=-\Hmin^\eps(A|CI_2)_\rho + \log N \\
    &\leq-\big(\Hmin^\eps(A|C)_\rho - \log N\big) + \log N\\
    &=-\Hmin^\eps(A|C)_\rho + 2 \log N.
  \end{align*}
  Now, observe that $\rho_{AC} = \frac1N \sum_i \rho_{AC}^i$. Thus, \cref{lem:min entropy mixture} can be applied, and we obtain
  \begin{align*}
    &\quad -\Hmin^\eps(A|C)_\rho \leq -\min_i \Hmin^\eps(A|C)_{\rho^i} \nonumber \\
    &= \max_{i} -\Hmin^\eps(A|C)_{\rho^i} = \max_{i} \Hmax^\eps(A|B)_{\rho^i}
  \end{align*}
  by another application of~\eqref{eq:smooth max entropy duality}.
\end{proof}

\section{Quantum capacity of compound quantum channels with informed sender}\label{app:plain IS}

\begin{lemma}\label{thm:plain one-shot}
Let $\Pi_{A\to B} = \{\mathcal N^i_{A\to B}\}_{i=1}^N$ be a finite compound channel, $\{\rho^i_{AA'}\}_{i=1}^N$ pure states, $M_0$ an integer such that $M_0 \leq d_A$, and $\eps \in (0,1]$. Then, there exist quantum operations $\mathcal E^i_{A_0\to A}$ and $\mathcal D_{B\to A_0}$ with $d_{A_0} = M_0$, such that
  \begin{align*}
  &\quad \min_i F(\Phi^+_{A_0R},\mathcal D \circ \mathcal N^i \circ \mathcal E^i(\Phi^+_{A_0R})) \nonumber \\
  &\geq 1 - 16 N (N+2) \bigl(\sqrt{\delta} + 6 \sqrt{\eps} \bigr),
  \end{align*}
  where
  \begin{align*}
    	\delta &= \max_i 2^{-\frac12 \bigl( -\Hmax^\eps(A'|B)_{\mathcal N^i(\rho^i)} - \log M_0  - 2 \log N \bigr)}.
  \end{align*}
\end{lemma}

By standard arguments this one-shot coding result can be lifted to an asymptotic IID capacity formula in terms of the regularized coherent information (as first derived and shown optimal for compound quantum channels in~\cite{bjelakovic2008quantum,bjelakovic09,bjelakovic2009entanglement}).

\begin{proof}[Proof of \cref{thm:plain one-shot}]
Denote by $W^i_{ A\rightarrow BE}$ a Stinespring dilation of $\mathcal N^i_{ A\to B}$ and by $\mathcal N^{i, c}_{A\to E}$ the corresponding complementary channel. We define the new channel
\begin{align*}
  \mathcal T_{AI\rightarrow B}(\cdot):= \sum_{i=1}^N \mathcal N^i_{ A\rightarrow B}[\braket{i|\cdot|i}],
\end{align*}
and input $\rho_{AI} := \frac1N \sum_i \rho^i_{A} \otimes \ket i\!\!\bra i$ with purification $\ket{\rho_{AIA'I'}} := \frac1{\sqrt N} \sum_i \ket{\rho^i_{AA'}} \otimes \ket{ii}$. We find $\mathcal T_{AI\rightarrow B}[\rho_{AI}]=\frac1N \sum_i \mathcal N^i_{A\rightarrow B}[\rho^i_{A}]$. We obtain a Stinespring dilation of the channel $\mathcal T_{AI\rightarrow B}$ by
\begin{align*}
 W_{AI \rightarrow BEI}(\cdot):=\sum_i \ket i\!\!\bra i \otimes W^i_{A\rightarrow BE}(\cdot).
\end{align*}
Moreover, we define the completely positive map
\begin{align*}
  &\mathcal E_{A_0I \to AI}(\sigma_{A_0I}) := d_{AI} O_{AI}(\rho) U_{AI} J_{A_0\to A} \\
  &\qquad \sigma_{A_0I} (J_{A_0\to A})^\dagger (U_{AI})^\dagger (O_{AI}(\rho))^\dagger,
\end{align*}
where the $J_{A_0\to A}$ are isometries, the $U_{AI}$ denote unitaries that will later be chosen at random, and $O_{AI}(\cdot)$ is defined as in~\eqref{eq:OA definition}. The main idea is to apply the one-shot decoupling lemma from~\cite[Theorem 3.1]{dupuis2014one},
\begin{align*}
&\quad\EE\Big\lVert \tr_B \Big[W_{AI\to BEI}  \mathcal E_{A_0I\to AI} (\Phi^+_{A_0R} \otimes \ket 0\!\!\bra 0_I ) W_{AI\to BEI}^\dagger \Big]  \nonumber \\
&\quad - \omega_{EI} \otimes \tau_R \Big\rVert_1\notag\\
&\leq2^{-\frac12 \Big(\Hmin^\eps(AI|R)_\sigma - \Hmax^\eps(A'I'|B)_{\mathcal T (\rho)}\Big)}+12\eps,
\end{align*}
where $\omega_{EI}:= \tr_B[W_{AI \rightarrow BEI} \rho_{AI} W_{AI \rightarrow BEI}^\dagger]$ and $\sigma_{AIR}:=J_{A_0\rightarrow A} (\Phi^+_{A_0R} \otimes \ket0\!\!\bra0_I) J_{A_0\rightarrow A}^\dagger$. This can be simplified using the fact that $\Hmin^\eps(AI|R)_\sigma = -\log M_0$. Next, we write $\eta_{A'I'B}:= \frac1N \sum_i \mathcal N_{A\to B}^i(\rho^i_{AA'}) \otimes \ket i\!\!\bra i$ and apply \cref{lem:max entropy mixture} to carry out the following sequence of inequalities:
\begin{align*}
 -\Hmax^\eps(A'I'|B)_\eta&\geq -\Hmax^\eps(A'|B)_\eta \\
 &\geq -\max_i \Hmax^\eps(A'|B)_{\mathcal N^i(\rho^i)} - 2 \log N.
\end{align*}
This implies that
\begin{align*}
&\EE\Big\lVert \tr_B \Big[W_{AI\to BEI}  \mathcal E_{A_0I\to AI} (\Phi^+_{A_0R} \otimes \ket 0\!\!\bra 0_I ) W_{AI\to BEI}^\dagger \Big]  \nonumber \\
&\quad - \omega_{EI} \otimes \tau_R \Big\rVert_1\leq \delta + 12\eps.
\end{align*}
Now \cref{lem:uhlmann one is normalized} implies that there exists a quantum operation $\mathcal D_{B\to A_0}$ (partial isometry) such that
\begin{align}\label{eq:plain decouple 2}
&\EE\Big\lVert \mathcal D_{B\to A_0} \circ \mathcal T_{AI\rightarrow B} \circ \mathcal E_{A_0I \to AI} \Big(\Phi^+_{A_0R} \otimes \ket 0\!\!\bra 0_I \Big) -  \Phi^+_{A_0R} \Big\rVert_1 \nonumber \\
&\leq \delta+12\eps + 2\sqrt{2(\delta+12\eps)}.
\end{align}
On the other hand, we have
\begin{align*}
  &\quad \mathcal T_{AI\rightarrow B}\circ \mathcal E_{A_0I \to AI}(\Phi^+_{A_0R} \otimes \ket0\!\!\bra0_I)\\
  &= \frac1N \sum_i \mathcal N^i_{A\rightarrow B} \circ \widetilde{\mathcal E}^i_{A_0 \to A} (\Phi^+_{A_0R}),
\end{align*}
where we abbreviated
\begin{align*}
&\widetilde {\mathcal E}^i_{A_0\to A}(\sigma_{A_0}):= d_A O_A\left(\rho^i\right) \sqrt{N} \bra{i}U_{AI}\ket{0}  J_{A_0\to A} \sigma_{A_0} \\
&\qquad(J_{A_0\to A})^\dagger \sqrt{N} \bra{0}U_{AI}^\dagger\ket{i}  O_A\left(\rho^i\right)^\dagger.
\end{align*}
Now, the one-shot decoupling lemma from~\cite[Theorem 3.1]{dupuis2014one} gives 
\begin{align*}
&\quad\EE\Big\lVert \tr_A \Big[\widetilde{\mathcal E}^i_{A_0\to A}(\Phi^+_{A_0R})\Big]  - \tau_R \Big\rVert_1 \\
&\leq2^{-\frac12\Big(\Hmin^\eps(A'I')_{\rho^i} - \log M_0 \Big)} +12\eps,
\end{align*}
where we have introduced $\rho^i_{A'I'} = \rho^i_{A'}\otimes \ket i\!\!\bra i_{I'}$. We now note that
\begin{align*}
\Hmin^\eps(A'I')_{\rho^i} &= \Hmin^\eps(A')_{\rho^i} \nonumber \\
&\geq  \Hmin^\eps(A'|E)_{\mathcal N^{i,c}(\rho^i)} =  -\Hmax^\eps(A'|B)_{\mathcal N^i(\rho^i)}
\end{align*}
which implies that
\begin{align}\label{eq:plain decouple 3.1}
\EE\Big\lVert \tr_A \Big[\widetilde{\mathcal E}^i_{A_0\to A}(\Phi^+_{A_0R})\Big]  - \tau_R \Big\rVert_1\leq\delta +12\eps.
\end{align}
Combining~\eqref{eq:plain decouple 2}, \eqref{eq:plain decouple 3.1}, Markov's inequality $\Prob(Z > k \EE(Z)) \leq 1/k$ (for $k=N+2$), and the union bound, we find that there exist a unitary $U_{AI}$ and a quantum operation $\mathcal D_{B\to A_0}$ such that
\begin{align*}
  \left\lVert \tr_A\left[\widetilde{\mathcal E}^i_{A_0\to A}(\Phi^+_{A_0R})\right] - \tau_{R} \right\rVert_1
  \leq (N+2) (\delta + 12\eps)
\end{align*}
for $i=1,\dots,N$,
with
\begin{align*}
		&\quad\left\lVert \mathcal D_{B\to A_0} \circ \frac1N \sum_{i=1}^N \mathcal N^i_{A\to B} \circ \widetilde{\mathcal E}^i_{A_0\to A}(\Phi^+_{A_0R}) - \Phi^+_{A_0R} \right\rVert_1 \\
		&\leq (N+2)(\delta+12\eps + 2\sqrt{2(\delta+12\eps})).
  \end{align*}
As a consequence of the first bound and \cref{lem:uhlmann one is normalized}, we can then find quantum operations $\{ \mathcal E^i_{A_0\to A} \}_{i=1}^N$ (partial isometries) such that
  \begin{align*}
  &\quad\left\lVert (\widetilde{\mathcal E}^i_{A_0\to A} - \mathcal E^i_{A_0\to A})(\Phi^+_{A_0R}) \right\rVert_1 \nonumber \\
  &\leq (N+2) (\delta + 12\eps) + 2 \sqrt{2 (N+2) (\delta + 12\eps)}.
  \end{align*}
Using the triangle inequality as well as the fact that $\mathcal D$ and the $\mathcal N^i$ are completely positive and trace-nonincreasing, we obtain that
  \begin{align*}
  	&\quad\left\lVert\bigl(\mathcal D_{B\to A_0} \circ  \frac1N \sum_{i=1}^N\mathcal N^i_{A\to B} \circ \mathcal E^i_{A_0\to A}\bigr)(\Phi^+_{A_0R}) - \Phi^+_{A_0R} \right\rVert_1 \\
  	&\leq\left\lVert\bigl(\mathcal D_{B\to A_0} \circ \frac1N \sum_{i=1}^N \mathcal N^i_{A\to B} \circ \widetilde{\mathcal E}^i_{A_0A_1\to A}\bigr)(\Phi^+_{A_0R}) - \Phi^+_{A_0R}\right\rVert_1  \\
	&\quad + \frac1N \sum_{i=1}^N\left\lVert(\widetilde{\mathcal E}^i_{A_0\to A} - \mathcal E^i_{A_0\to A})(\Phi^+_{A_0R})\right\rVert_1 \\
  	&\leq(N+2) (\delta + 12\eps) + 2 \sqrt{2 (N+2) (\delta + 12\eps)}  \\
	&\quad +  (N+2)(\delta+12\eps + 2\sqrt{2(\delta+12\eps}))\\
	&\leq 8 (N+2) (\sqrt{\delta} + 6\sqrt{\eps}).
\end{align*}
For the last step we assumed that $\delta \in (0,1]$ (without loss of generality). Lastly, we use~\eqref{eq:fidelity lower bound one is normalized} to turn this into a lower bound on the entanglement fidelity:
  \begin{align*}
  &\quad F(\Phi^+_{A_0R}, \bigl( \frac1N \sum_{i=1}^N \mathcal D_{B\to A_0} \circ \mathcal N^i_{A\to B} \circ \mathcal E^i_{A_0\to A}\bigr)(\Phi^+_{A_0R} )) \\
  &\geq 1 - 16 (N+2) \bigl(\sqrt{\delta} + 6 \sqrt{\eps} \bigr).
  \end{align*}
Using the same argument that was used to derive~\eqref{eq:entanglement fidelity union bound} this implies the claim.
\end{proof}


\bibliographystyle{IEEEtran}
\bibliography{IEEEabrv,literature}

\begin{thebibliography}{10}
\providecommand{\url}[1]{#1}
\csname url@samestyle\endcsname
\providecommand{\newblock}{\relax}
\providecommand{\bibinfo}[2]{#2}
\providecommand{\BIBentrySTDinterwordspacing}{\spaceskip=0pt\relax}
\providecommand{\BIBentryALTinterwordstretchfactor}{4}
\providecommand{\BIBentryALTinterwordspacing}{\spaceskip=\fontdimen2\font plus
\BIBentryALTinterwordstretchfactor\fontdimen3\font minus
  \fontdimen4\font\relax}
\providecommand{\BIBforeignlanguage}[2]{{%
\expandafter\ifx\csname l@#1\endcsname\relax
\typeout{** WARNING: IEEEtran.bst: No hyphenation pattern has been}%
\typeout{** loaded for the language `#1'. Using the pattern for}%
\typeout{** the default language instead.}%
\else
\language=\csname l@#1\endcsname
\fi
#2}}
\providecommand{\BIBdecl}{\relax}
\BIBdecl

\bibitem{shannon48}
C.~Shannon, ``A mathematical theory of communication,'' \emph{Bell System
  Technical Journal}, vol.~27, pp. 379--423, 1948.

\bibitem{bennett02}
C.~H. Bennett, P.~W. Shor, J.~A. Smolin, and A.~Thapliyal,
  ``Entanglement-assisted capacity of a quantum channel and the reverse
  {S}hannon theorem,'' \emph{IEEE Transactions on Information Theory}, vol.~48,
  no.~10, pp. 2637--2655, 2002.

\bibitem{matthews2014finite}
W.~Matthews and S.~Wehner, ``Finite blocklength converse bounds for quantum
  channels,'' \emph{IEEE Transactions on Information Theory}, vol.~60, no.~11,
  pp. 7317--7329, 2014.

\bibitem{lapidoth98}
A.~Lapidoth and P.~Narayan, ``Reliable communication under channel
  uncertainty,'' \emph{IEEE Transactions on Information Theory}, vol.~44,
  no.~6, pp. 2148--2177, 1998.

\bibitem{blackwell59}
D.~Blackwell, L.~Breiman, and A.~J. Thomasian, ``The capacity of a class of
  channels,'' \emph{Annals of Mathematical Statistics}, vol.~30, no.~4, pp.
  1229--1241, 1959.

\bibitem{wolfowitz59}
J.~Wolfowitz, ``Simultaneous channels,'' \emph{Archive for Rational Mechanics
  and Analysis}, vol.~4, no.~1, pp. 371--386, 1959.

\bibitem{wolfowitz78}
------, \emph{Coding Theorems of Information Theory}, 3rd~ed.\hskip 1em plus
  0.5em minus 0.4em\relax Springer-Verlag, New York, 1978.

\bibitem{Shrader09}
B.~Shrader and H.~Permuter, ``Feedback capacity of the compound channel,''
  \emph{IEEE Transactions on Information Theory}, vol.~55, no.~8, pp.
  3629--3644, 2009.

\bibitem{holevo98}
A.~Holevo, ``The capacity of the quantum channel with general signal states,''
  \emph{IEEE Transactions on Information Theory}, vol.~44, no.~1, pp. 269--273,
  1998.

\bibitem{schumacher97}
B.~Schumacher and M.~Westmoreland, ``Sending classical information via noisy
  quantum channels,'' \emph{Physical Review A}, vol.~56, no.~1, pp. 131--138,
  1997.

\bibitem{devetak05b}
I.~Devetak, ``The private classical capacity and quantum capacity of a quantum
  channel,'' \emph{IEEE Transactions on Information Theory}, vol.~51, no.~1,
  pp. 44--55, 2005.

\bibitem{lloyd97}
S.~Lloyd, ``The capacity of the noisy quantum channel,'' \emph{Physical Review
  A}, vol.~55, no.~3, pp. 1613--1622, 1996.

\bibitem{shor02}
P.~W. Shor, ``The quantum channel capacity and coherent information,'' in
  \emph{Lectures Notes, MSRI Workshop on Quantum Computation}, 2002.

\bibitem{bowen02}
G.~Bowen, ``Quantum feedback channels,'' \emph{IEEE Transactions on Information
  Theory}, vol.~50, no.~10, pp. 2429--2434, 2004.

\bibitem{matthews14}
D.~Leung and W.~Matthews, ``On the power of ppt-preserving and non-signalling
  codes,'' \emph{IEEE Transactions on Information Theory}, vol.~61, no.~8, pp.
  4486--4499, 2015.

\bibitem{PhysRevLett.70.1895}
C.~H. Bennett, G.~Brassard, C.~Cr\'epeau, R.~Jozsa, A.~Peres, and W.~K.
  Wootters, ``Teleporting an unknown quantum state via dual classical and
  {E}instein-{P}odolsky-{R}osen channels,'' \emph{Physical Review Letters},
  vol.~70, pp. 1895--1899, 1993.

\bibitem{PhysRevLett.69.2881}
C.~H. Bennett and S.~J. Wiesner, ``Communication via one- and two-particle
  operators on {E}instein-{P}odolsky-{R}osen states,'' \emph{Physical Review
  Letters}, vol.~69, pp. 2881--2884, 1992.

\bibitem{boche09}
I.~Bjelakovi{\'c} and H.~Boche, ``Classical capacities of compound and averaged
  quantum channels,'' \emph{IEEE Transactions on Information Theory}, vol.~55,
  no.~7, pp. 3360--3374, 2009.

\bibitem{hayashi09b}
M.~Hayashi, ``Universal coding for classical-quantum channel,''
  \emph{Communications in Mathematical Physics}, vol. 289, no.~3, pp.
  1087--1098, 2009.

\bibitem{Boche13}
I.~Bjelakovi{\'c}, H.~Boche, G.~Jan{\ss}en, and J.~N\"{o}tzel, ``Arbitrarily
  varying and compound classical-quantum channels and a note on quantum
  zero-error capacities,'' in \emph{Information Theory, Combinatorics, and
  Search Theory}, ser. Lecture Notes in Computer Science Volume, vol.
  7777.\hskip 1em plus 0.5em minus 0.4em\relax Springer, 2013, pp. 247--283.

\bibitem{bjelakovic2008quantum}
I.~Bjelakovi\ifmmode~\acute{c}\else \'{c}\fi{}, H.~Boche, and J.~N\"{o}tzel,
  ``Quantum capacity of a class of compound channels,'' \emph{Physical Review
  A}, vol.~78, no.~4, p. 042331, 2008.

\bibitem{bjelakovic09}
I.~Bjelakovi{\'c}, H.~Boche, and J.~N\"{o}tzel, ``{Entanglement transmission
  capacity of compound channels},'' in \emph{Proc. IEEE ISIT 2009}, 2009, pp.
  1889--1893.

\bibitem{bjelakovic2009entanglement}
------, ``Entanglement transmission and generation under channel uncertainty:
  Universal quantum channel coding,'' \emph{Communications in Mathematical
  Physics}, vol. 292, no.~1, pp. 55--97, 2009.

\bibitem{1751-8121-40-28-S20}
N.~Datta and T.~C. Dorlas, ``The coding theorem for a class of quantum channels
  with long-term memory,'' \emph{Journal of Physics A: Mathematical and
  Theoretical}, vol.~40, no.~28, p. 8147, 2007.

\bibitem{Datta2008}
N.~Datta, Y.~Suhov, and T.~C. Dorlas, ``Entanglement assisted classical
  capacity of a class of quantum channels with long-term memory,''
  \emph{Quantum Information Processing}, vol.~7, no.~6, pp. 251--262, 2008.

\bibitem{horodecki05}
M.~Horodecki, J.~Oppenheim, and A.~Winter, ``Partial quantum information,''
  \emph{Nature}, vol. 436, no. 7051, pp. 673--6, 2005.

\bibitem{horodecki06}
------, ``Quantum state merging and negative information,''
  \emph{Communications in Mathematical Physics}, vol. 269, no.~1, pp. 107--136,
  2006.

\bibitem{hayden08}
P.~Hayden, M.~Horodecki, J.~Yard, and A.~Winter, ``A decoupling approach to the
  quantum capacity,'' \emph{Open Systems {\&} Information Dynamics}, vol.~15,
  no.~01, p.~7, 2008.

\bibitem{abeyesinghe09}
A.~Abeyesinghe, I.~Devetak, P.~Hayden, and A.~Winter, ``The mother of all
  protocols: Restructuring quantum information's family tree,''
  \emph{Proceedings of the Royal Society A}, vol. 465, no. 2108, pp.
  2537--2563, 2009.

\bibitem{dupuis2009decoupling}
F.~Dupuis, ``The decoupling approach to quantum information theory,'' Ph.D.
  dissertation, Universit\'{e} de Montr\'{e}al, 2009.

\bibitem{dupuis2014one}
F.~Dupuis, M.~Berta, J.~Wullschleger, and R.~Renner, ``One-shot decoupling,''
  \emph{Communications in Mathematical Physics}, vol. 328, no.~1, pp. 251--284,
  2014.

\bibitem{renner2005security}
R.~Renner, ``Security of quantum key distribution,'' Ph.D. dissertation, ETH
  Zurich, 2005.

\bibitem{tomamichel2015quantum}
M.~Tomamichel, \emph{Quantum Information Processing with Finite
  Resources}.\hskip 1em plus 0.5em minus 0.4em\relax Springer, 2016.

\bibitem{Kretschmann04}
D.~Kretschmann and R.~F. Werner, ``Tema con variazioni: quantum channel
  capacity,'' \emph{New Journal of Physics}, vol.~6, no.~1, p.~26, 2004.

\bibitem{Alicki04}
R.~Alicki and M.~Fannes, ``Continuity of quantum conditional information,''
  \emph{Journal of Physics A: Mathematical and General}, vol.~37, no.~5, p.
  L55, 2004.

\bibitem{winter15}
A.~Winter, ``Tight uniform continuity bounds for quantum entropies: Conditional
  entropy, relative entropy distance and energy constraints,''
  \emph{Communications in Mathematical Physics}, vol. 347, no.~1, pp. 291--313,
  2016.

\bibitem{boche16paper}
H.~Boche, G.~Janssen, and S.~Kaltenstadler, ``Entanglement-assisted classical
  capacities of compound and arbitrarily varying quantum channels,''
  \emph{Quantum Information Processing (to appear)}, 2017.

\bibitem{boche16}
------, ``Entanglement assisted classical capacity of compound quantum
  channels,'' in \emph{Proc. IEEE ISIT 2016}, 2016, pp. 1680--1684.

\bibitem{wildebook13}
M.~M. Wilde, \emph{Quantum Information Theory}.\hskip 1em plus 0.5em minus
  0.4em\relax Cambridge University Press, 2013.

\bibitem{fawzi2015quantum}
O.~Fawzi and R.~Renner, ``Quantum conditional mutual information and
  approximate {M}arkov chains,'' \emph{Communications in Mathematical Physics},
  vol. 340, no.~2, pp. 575--611, 2015.

\bibitem{shirokov15}
M.~Shirokov, ``Tight continuity bounds for the quantum conditional mutual
  information, for the {H}olevo quantity and for capacities of quantum
  channels,'' 2015.

\bibitem{tomamichel09}
M.~Tomamichel, R.~Colbeck, and R.~Renner, ``Duality between smooth min- and
  max-entropies,'' \emph{IEEE Transactions on Information Theory}, vol.~56,
  no.~9, pp. 4674--4681, 2010.

\bibitem{koenig08}
R.~K{\"{o}}nig, R.~Renner, and C.~Schaffner, ``The operational meaning of min-
  and max-entropy,'' \emph{IEEE Transactions on Information Theory}, vol.~55,
  no.~9, pp. 4337--4347, 2009.

\bibitem{tomamichel2009fully}
M.~Tomamichel, R.~Colbeck, and R.~Renner, ``A fully quantum asymptotic
  equipartition property,'' \emph{IEEE Transactions on Information Theory},
  vol.~55, no.~12, pp. 5840--5847, 2009.

\bibitem{adami97}
C.~Adami and N.~Cerf, ``Von {N}eumann capacity of noisy quantum channels,''
  \emph{Physical Review A}, vol.~56, no.~5, pp. 3470--3483, 1997.

\bibitem{bhatia2013matrix}
R.~Bhatia, \emph{Matrix analysis}.\hskip 1em plus 0.5em minus 0.4em\relax
  Springer-Verlag New York, 2013.

\bibitem{ahlswede10}
R.~Ahlswede, I.~Bjelakovi{\'c}, H.~Boche, and J.~N\"{o}tzel, ``{Entanglement
  transmission over arbitrarily varying quantum channels},'' in
  \emph{Proc.~IEEE ISIT 2010}, 2010, pp. 2718--2722.

\bibitem{Ahlswede2012}
R.~Ahlswede, I.~Bjelakovi{\'{c}}, H.~Boche, and J.~N{\"o}tzel, ``Quantum
  capacity under adversarial quantum noise: Arbitrarily varying quantum
  channels,'' \emph{Communications in Mathematical Physics}, vol. 317, no.~1,
  pp. 103--156, 2012.

\bibitem{noetzel13}
H.~Boche and J.~N{\"o}tzel, ``Arbitrarily small amounts of correlation for
  arbitrarily varying quantum channels,'' \emph{Journal of Mathematical
  Physics}, vol.~54, no.~11, p. 112202, 2013.

\bibitem{blackwell60}
D.~Blackwell, L.~Breiman, and A.~J. Thomasian, ``The capacities of certain
  channel classes under random coding,'' \emph{Annals of Mathematical
  Statistics}, vol.~31, pp. 558--567, 1960.

\bibitem{stiglitz66}
I.~G. Stiglitz, ``Coding for a class of unknown channels,'' \emph{IEEE
  Transactions on Information Theory}, vol.~12, pp. 189--195, 1966.

\bibitem{ahlswede69}
R.~Ahlswede and J.~Wolfowitz, ``Correlated decoding for channels with
  arbitrarily varying channel probability functions,'' \emph{Information and
  Control}, vol.~14, pp. 457--473, 1969.

\end{thebibliography}





\end{document}